\documentclass{article}

\usepackage[final]{neurips_2024}

\usepackage[utf8]{inputenc} % allow utf-8 input
\usepackage[T1]{fontenc}    % use 8-bit T1 fonts
\usepackage{hyperref}       % hyperlinks
\usepackage{url}            % simple URL typesetting
\usepackage{booktabs}       % professional-quality tables
\usepackage{amsfonts}       % blackboard math symbols
\usepackage{nicefrac}       % compact symbols for 1/2, etc.
\usepackage{microtype}      % microtypography
\usepackage{xcolor}         % colors
\usepackage{subcaption}
\usepackage{nicematrix}

\usepackage{float}
\usepackage{amssymb}
\usepackage{amsmath}
\usepackage{dsfont}
\usepackage{amsthm}
\usepackage{bm}
\usepackage{mathtools}
\usepackage{xcolor}
\usepackage{xspace}
\usepackage[toc]{appendix}
\newtheorem{assumption}{Assumption}
\newtheorem{definition}{Definition}
\newtheorem{lemma}{Lemma}
\newtheorem{theorem}{Theorem}
\newtheorem{remark}{Remark}
\newtheorem{example}{Example}
\newtheorem{corollary}{Corollary}

\usepackage{paralist}

\newcommand{\R}{\mathbb R}

\makeatletter
\newcommand{\vast}{\bBigg@{4}}
\newcommand{\Vast}{\bBigg@{5}}
\makeatother

\title{Identification of Analytic Nonlinear Dynamical Systems with Non-asymptotic Guarantees}

\author{%
  Negin Musavi \\
  \texttt{nmusavi2@illinois.edu} \\
  \And
  Ziyao Guo \\
  \texttt{ziyaog2@illinois.edu} \\
  \AND
  Geir Dullerud \\
  \texttt{dullerud@illinois.edu} \\
  \And
  Yingying Li \\
  \texttt{yl101@illinois.edu} \\
  \And
  {\normalfont Coordinated Science Laboratory}\\
  University of Illinois Urbana-Champaign%\\
}

\begin{document}

\maketitle

\begin{abstract}
This paper focuses on the system identification of an important class of nonlinear systems: linearly parameterized nonlinear systems, which enjoys wide applications in robotics and other mechanical systems. We consider two system identification methods: least-squares estimation (LSE), which is a point estimation method; and set-membership estimation (SME), which estimates an uncertainty set that contains the true parameters. We provide non-asymptotic convergence rates for LSE and SME under  i.i.d. control inputs and control policies with i.i.d. random perturbations, both of which are considered as non-active-exploration inputs. Compared with the counter-example based on piecewise-affine systems in the literature, the success of non-active exploration in our setting relies on a key assumption on the system dynamics: we require the system functions to be real-analytic. Our results, together with the piecewise-affine counter-example, reveal the importance of differentiability in nonlinear system identification through non-active exploration. Lastly, we numerically compare our theoretical bounds with the empirical performance of LSE and SME on a pendulum example and a quadrotor example.
\end{abstract}

\section{Introduction}\label{sec:intro}Learning control-dynamical systems with statistical methodology has received significant attention in the past decade \citep{sarker2023accurate,li2023non,chen2021black,simchowitz2020naive,wagenmaker2020active,simchowitz2018learning,dean2018regret,abbasi2011regret,li2021distributed}. In particular, the estimation of linear dynamical systems, e.g. $x_{t+1}=A^* x_t +B^* u_t+w_t$, is relatively well-studied: it has been shown that non-active exploration by i.i.d. noises on control inputs $u_t$ and system disturbances $w_t$ are already enough for accurate system identification, and least square estimation (LSE) can achieve the optimal estimation convergence rate \citep{simchowitz2020naive,simchowitz2018learning}.  

However, nonlinear control systems are ubiquitous in real-world applications, e.g. robotics \citep{siciliano2010robotics,alaimo2013mathematical}, power systems \citep{simpson2016voltage}, transportation \citep{kong2015kinematic}, etc. Motivated by this, there has been a lot of attention on learning nonlinear systems recently. One natural and popular direction to study nonlinear system identification is on learning linearly parameterized nonlinear systems as defined below, which is a straightforward extension from the standard linear systems \citep{mania2022active,khosravi2023representer,foster2020learning}
\begin{equation*}\label{equ: linearly parametrized}
	x_{t+1}=\theta_* \phi(x_t, u_t)+w_t
\end{equation*}
where $\theta_*$ is a vector of unknown parameters and $\phi(x_t, u_t)$ is a known vector of  nonlinear features.

On the one hand, some classes of
these systems
% \eqref{equ: linearly parametrized}
are shown to enjoy similar benefits of linear systems. For example, bilinear systems can also be estimated by LSE under 
non-active exploration with i.i.d. noises~\citep{sattar2022finite}, as well as linear systems with randomly perturbed nonlinear policies \citep{li2023non}. 

On the other hand, it is also known that non-active exploration is insufficient for general linearly parameterized nonlinear systems. In particular, ~\citep{mania2022active} provides a counter example showing that non-active exploration is insufficient to learn accurate models under piece-wise affine feature functions. This motivates a sequence of follow-up work on the design of active exploration for nonlinear system estimation, which is largely motivated by the non-smooth feature functions such as ReLu in neural networks~\citep{mania2022active,kowshik2021near,khosravi2023representer}.

However, there is a big gap between bilinear systems, which is infinitely differentiable,  and the counter example by  non-smooth systems. A natural question is: to what extent can non-active exploration still work for linearly parameterized nonlinear systems?

\paragraph{Contributions.} One major contribution of this paper is showing that LSE with non-active i.i.d. noises can efficiently learn any linearly parameterized nonlinear systems with real-analytic feature functions and provide a non-asymptotic convergence rate. Notice that real-analytic feature functions are common in physical systems. For example,  polynomial systems satisfy this requirement and have wide applications in power systems~\citep{simpson2016voltage}, fluid dynamics \citep{noack2003hierarchy}, etc.  Further, trigonometric functions also satisfy the real-analytic property so a large range of robotics and mechanical systems also satisfy this requirement \citep{siciliano2010robotics,alaimo2013mathematical}.

A side product of our LSE convergence rate analysis is the convergence rate for another commonly used uncertainty quantification method in control: set membership estimation (SME).

Numerically, we test our theoretical results in pendulum and quadrotor systems. Simulations show that LSE and SME can indeed efficiently explore the system and converge to the true parameter under non-active exploration noises.

Technically, the key step in our proof is establishing the block-martingale-small-ball condition (BMSB) for general analytic feature functions, which greatly generalizes the bilinear feature function in~\cite{sattar2022finite}. Our result is built on an intuition inspired by the counter example in \citep{mania2022active}: the counter example in \citep{mania2022active} requires that some feature function is zero in a certain region, so nothing can be learned about its parameter if the states stay in this region. However, analytic functions cannot be a constant zero in a positive-measure region unless it is a constant zero everywhere. Therefore, the counter example does not work, and non-active exploration around any states can provide some useful information. Our proof formalizes this intuition by utilizing the Paley-Zygmund Petrov inequality~\citep{petrov2007lower}.

\paragraph{Related work.} Inspired by neural network parameterization, nonlinear systems of the form $x_{t+1}=\phi(A_* x_t)+w_t$ is also studied in the literature, where $\phi(\cdot)$ is a known nonlinear link function and $A_*$ is unknown. The least square cost is no longer quadratic or even convex in this case and various optimization methods have been proposed to learn this type of systems \citep{kowshik2021near,sattar2022finite,foster2020learning}.

Another related  line of research focuses on nonlinear regression with dependent data \citep{ziemann2022learning,ziemann2023tutorial,ziemann2024sharp},\footnote{$y_t=f_*(x_t)+w_t$ is considered, where $x_t$ and $y_t$ correlate with the historical data.} which can be applied to nonlinear system identification. The nonlinear regression in \citep{ziemann2022learning,ziemann2023tutorial,ziemann2024sharp} is based on non-parametric LSE and its variants, and their convergence rates under different scenarios have been analyzed. It is interesting to note that this line of work usually assumes certain persistent excitation assumptions,\footnote{For example, \citep{ziemann2022learning} assumes hyper-contractivity, and \citep{ziemann2024sharp} assumes the empirical covariance of the $\{x_t\}_{t\geq 0}$ process is invertible with high probability (Corollary 3.2).} whereas our paper demonstrates that persistent excitation holds by establishing the BMSB condition for linearly parameterized and real-analytic nonlinear control systems.

Uncertainty set estimation  is crucial for robust control under model uncertainties \cite{lu2023robust,lorenzen2019robust,li2021online}. SME is a widely adopted uncertainty set estimation method in robust adaptive control \citep{lorenzen2019robust,lu2023robust,bertsekas1971control,bai1995membership}. Recently, there is an emerging interest in analyzing SME's convergence and convergence rate for dynamical systems \citep{li2024icml,lu2019robust,xu2024convergence}, because previous analysis focus more on the linear regression problem (e.g. \citep{akccay2004size,bai1998convergence}). There are also recent applications of SME to online control \cite{yu2023online}, power systems \cite{yeh2024online}, and computer vision \cite{gao2024closure,tang2024uncertainty}.

\textbf{Notation.}
The set of non-negative real numbers is denoted by  $\mathbb{R}_{\geq 0}$. The notation $\lceil \cdot \rceil$ stands for the ceiling function. For a real vector $z\in\mathbb{R}^{n}$, $\|z\|_{2}$ represents its $\ell_{2}$ norm, $\|z\|_{\infty}$ represents its $\ell_{\infty}$ norm, and  $z^{i}$  represents its $i$-th component with $i=1 \cdots n$. The set of real symmetric matrices is denoted by $\mathbb{S}^{n}$. For a real matrix $Z$, $Z^{\intercal}$ represents its transpose, $\|Z\|_{2}$ its maximum singular value, $\|Z\|_{F}$ its Frobenius norm, $\sigma_{\min}(Z)$ its minimum singular value, $\hbox{vec}(Z)$ its vectorization obtained by stacking its columns, and for a real square matrix $Z$, $\hbox{tr}(Z)$ represents its trace.  For a real symmetric matrix $Z$, $Z \succ 0$ and $Z \succeq 0$ indicate that $Z$ is positive definite and positive semi-definite, respectively.  For a measurable set $\mathcal{E} \subset \mathbb{R}^{n}$, $\lambda^{n}(\mathcal{E})$ represents its Lebesgue measure in $\mathbb{R}^{n}$ and $\mathcal{E}^{c}$ represents its complement in $\mathbb{R}^{n}$. The notation $\emptyset$ stands for the empty set. For a set $\mathcal{T}$ of matrices $\theta \in \mathbb{R}^{n \times m}$, $\hbox{diam}(\mathcal{T})$ denotes its diameter and it is defined as $\hbox{diam}(\mathcal{T}) = \sup_{\theta, \theta' \in \mathcal{T}} \|\theta-\theta'\|_{F}$. For $z_{i} \in \mathbb{R}$ with $i=1,\cdots,\ell$, the notation $\hbox{diag}(z_{1}, \cdots, z_{\ell})$ denotes a matrix in $\mathbb{R}^{\ell \times \ell}$ with diagonal entries of $z_{i}$. This paper uses $\texttt{truncated-Gaussian}(0, \sigma_{w}, [-w_{\max}, w_{\max}])$ to refer to the
truncated Gaussian distribution generated by Gaussian distribution with zero mean and $\sigma_{w}^{2}$ variance with truncated range $[-w_{\max}, w_{\max}]$. The same applies to multi-variate truncated Gaussian distributions.

\section{Problem Formulation and Preliminaries}\label{sec:prob}

This paper studies the system identification/estimation of linearly parameterized nonlinear systems:
\begin{equation}\label{eq:sys}
\begin{aligned}
    x_{t+1} = \theta_{*}\phi\big(x_{t},u_{t}\big) + w_{t},
\end{aligned}    
\end{equation}
where  $x_{t}\in\mathbb{R}^{n_{x}}$, $u_{t}\in\mathbb{R}^{n_{u}}$, and $w_{t}\in\mathbb{R}^{n_{x}}$ denote the state, control input, and system disturbance respectively;  $\theta_*\in \R^{n_x \times n_{\phi}}$ denotes the unknown parameters to be estimated, and $\phi(\cdot)$ denotes a vector of known nonlinear feature/basis functions, i.e., $\phi(\cdot)=(\phi^{1}(\cdot), \cdots, \phi^{n_{\phi}}(\cdot))^{\intercal}$, where $\phi^{i}(\cdot):\mathbb{R}^{n_{x} + n_{u}} \rightarrow \mathbb{R}$. Without loss of generality, we consider zero initial condition, i.e., $x_0=0$, and linearly independent feature functions, that is, $\sum_{i=1}^{n_{\phi}} c_{i}\phi^{i}(x_{t},u_{t}) = 0$ implies that $c_i=0$ for all $i$.\footnote{If the features are not independent, they can be converted to independent ones since the features are known.}

The linearly parameterized nonlinear system \eqref{eq:sys} is a natural generalization of linear control systems $x_{t+1}=A_*x_t +B_*u_t+w_t$ and has wide applications in, e.g. robotics \citep{siciliano2010robotics,alaimo2013mathematical}, power systems \citep{simpson2016voltage}, transportation \citep{kong2015kinematic}, etc. Therefore, there has been a lot of research on learning this type of systems \eqref{eq:sys} by utilizing the methodology and insights from linear system estimation. For example, it is common to estimate a linearly parameterized nonlinear system by least square estimation (LSE), which enjoys desirable performance in linear systems.

In particular, LSE for \eqref{eq:sys} is reviewed below
\begin{align}\label{eq:ols-eq}
    \hat{\theta}_{T} = \underset{\hat{\theta}}{\arg\min} \sum_{t=0}^{T-1} \big\|x_{t+1} - \hat{\theta}\phi(x_{t}, u_{t}) \big\|_{2}^{2}.
\end{align}
For linear systems, LSE  enjoys the following good property:  LSE can achieve the optimal rate of convergence with i.i.d. noises $w_t$ and i.i.d. control inputs $u_t$ under proper conditions (~\citep{simchowitz2018learning}). This good property has been generalized to some linearly parameterized nonlinear systems, such as bilinear systems, and linear systems with nonlinear control policies.
Unfortunately,  general linearly parameterized nonlinear systems do not enjoy this good property of linear systems, meaning  i.i.d. random inputs may not provide enough exploration for non-smooth feature functions $\phi(\cdot)$. Therefore, a sequence of follow-up work focuses on the design of active exploration methods. 

However, due to the simplicity of  implementation, i.i.d. random inputs remain a popular  method in empirical research of system identification and  enjoy satisfactory performance sometimes, despite the lack of theoretical guarantees. Therefore, this paper aims to establish  more general conditions that allow provable convergence of  nonlinear system estimation  under i.i.d. random inputs.

In the rest of this paper, we will show that with certain smoothness and continuous conditions, i.i.d. random inputs are sufficient for estimation of \eqref{eq:sys}, which recovers the good property of linear systems.

\subsection{Assumptions}
In the following, we formally describe the smoothness and continuity conditions that enables efficient exploration of \eqref{eq:sys} by i.i.d. random inputs.

\begin{assumption}[Analytic feature functions]\label{ass: analytic}
All the components of feature vector $\phi(\cdot)$ are real analytic functions in $\R^{n_x+n_u}$,\footnote{This assumption can be relaxed to locally analytic functions in a large enough bounded set.} i.e., for every $1\leq i \leq n_{\phi}$, $\phi^i(x, u)$ is an infinitely differentiable function such that the Taylor expansion at every $(\bar x, \bar u)$ converges point-wisely to the $\phi^i(x,u)$ in a neighborhood of $(\bar x, \bar u)$.
\end{assumption}

Analytic functions include polynomial functions and trigonometric functions, which are important components of many physical systems in real-world applications, e.g. power systems, robotics, transportation systems, etc. In particular, we provide two illustrative examples below.

\begin{example}[Pendulum]\label{example: pendulum}
Many multi-link robotic manipulators can be understood as interconnected pendulum dynamics. The motion equations of a single pendulum, consisting of a mass $m$ suspended from a weightless rod of length $l$ fixed at a pivot with no friction, can be expressed as:
$$\ddot{\alpha} = - \dfrac{g}{l} \sin(\alpha) + \dfrac{u}{ml^{2}} + w,$$
where $\alpha$ represents the angle of the rod relative to the vertical axis, $g$ is the gravity constant, $u$ is the torque input, and $w$ is the disturbance applied to this system. After discretization the system dynamics can be rewritten in the structure of (\ref{eq:sys}) with the feature vector consisting of expressions involving $\sin(\alpha)$ and $u$, all of which are analytic functions. The matrix of unknown parameters contains terms of the pendulum's mass and the rod's length.
\end{example}

\begin{example}[Quadrotor \citep{alaimo2013mathematical}]\label{example: quadrotor}
Let $p \in \mathbb{R}^{3}$ and $v \in \mathbb{R}^{3}$ represent the center of mass position and velocity of the quadrotor in the inertial frame, respectively; let $\omega \in \mathbb{R}^{3}$ denote its angular velocity in the body-fixed frame, and $q \in \mathbb{R}^{4}$ denote the quaternion vector. The quadrotor's equations of motion can then be expressed as:
\begin{align*}
    \dfrac{d}{dt} \begin{pmatrix}
        p\\
        v\\
        q\\
        \omega
    \end{pmatrix} = \begin{pmatrix}
        v\\
        - g e_{z} + \frac{1}{m} Q f_{u}\\
        \frac{1}{2} \Omega q\\
        I^{-1} (\tau_{u} - \omega \times I \omega) 
    \end{pmatrix} + w,
\end{align*}
where $g$ is the gravity constant, $m$ is its total mass, $I = \emph{diag} (I_{xx}, I_{yy}, I_{zz})$ its inertia matrix with respect to the body-fixed frame, $f_{u} \in \mathbb{R}$ the total thrust, $\tau_{u} \in \mathbb{R}^{3}$ the total moment in the body-fixed frame, $e_{z} = (0, 0, 1)^{\intercal}$, $   Q = 
    \begin{pmatrix}
    q_{0}^{2} + q_{1}^{2} - q_{2}^{2} - q_{3}^{2} & 2(q_{1}q_{2} - q_{0}q_{3}) & 2(q_{0}q_{2} - q_{1}q_{3})\\
    2(q_{1}q_{2} - q_{0}q_{3}) & q_{0}^{2} - q_{1}^{2} + q_{2}^{2} - q_{3}^{2} & 2(q_{2}q_{3} - q_{0}q_{1})\\
    2(q_{1}q_{3} - q_{0}q_{2}) & 2(q_{0}q_{1} - q_{2}q_{3}) & q_{0}^{2} - q_{1}^{2} - q_{2}^{2} + q_{3}^{2}
\end{pmatrix},$ and 
$\Omega =    \begin{pmatrix}
     0 & -\omega_{1} & -\omega_{2} & -\omega_{3}\\
     \omega_{1} & 0 & \omega_{3} & -\omega_{2}\\
     \omega_{2} & -\omega_{3} & 0 & \omega_{1}\\
     \omega_{3} & \omega_{2} & -\omega_{1} & 0\\
    \end{pmatrix}$.

Similar to the pendulum example, after discretization the system dynamics can be rewritten in the structure of (\ref{eq:sys}) with the feature vector consisting of \textup{cubic polynomials} in states and inputs, which are real-analytic. The unknown parameters contain terms of the mass and inertial moments of the quadrotor.
\end{example}

Next, we introduce the assumption on the random inputs, which relies on the following definition.

\begin{definition}[Semi-continuous distribution]\label{def: semi continuous}
    We define a probability distribution $\mathbb{P}$ as semi-continuous  if there does not exist a set $E$ with Lebesgue measure zero such that $\mathbb{P}(E) = 1$.
\end{definition}

The semi-continuous distribution is a weaker requirement than  continuous distributions. In particular, any continuous distributions, or a mixture distribution with one component as a continuous distribution can satisfy the requirement of semi-continuity. {The semi-continuity can also be interpreted by the Lebesgue Decomposition Theorem (Chapter 6 of~\citep{halmos2013measure}) as discussed below. }

\begin{remark}[Connection with Lebesgue Decomposition Theorem]
Definition~\ref{def: semi continuous} can be interpreted by the Lebesgue Decomposition Theorem, which suggests that any probability distribution can be decomposed  into a purely atomic component and a non-atomic component (see more details in \cite{halmos2013measure}). 
A semi-continuous distribution as defined in Definition~\ref{def: semi continuous} requires  the distribution's non-atomic component to be nonzero. 
\end{remark}

In the following, we provide the assumptions on $w_t$ and $u_t$ using the semi-continuity definition.

\begin{assumption}[Bounded i.i.d. and semi-continuous disturbance]\label{ass: bounded iid semi continuous noises}
$w_{t}$ is i.i.d. following a semi-continuous distribution with zero
mean and a positive definite covariance matrix $\Sigma_w \succeq \sigma_{w}^{2}I_{n_{x}}\succ 0$ and a bounded support, i.e. $\|w_t\|_\infty \leq w_{\max}$ almost surely for all $t$.
\end{assumption}

The i.i.d. assumption is common in the literature of system identification for linear and nonlinear systems. As for the  bounded assumption on $w_t$, though being stronger than the sub-Gaussian assumption on $w_t$ in the literature of linear system estimation, it is a common assumption in the literature of  nonlinear system estimation \citep{mania2022active, shi2021meta, kimonline}. Further, in many physical applications, noises are usually bounded, e.g. the wind disturbances in quadrotor systems are bounded, the renewable energy injections in power systems are also bounded, etc.

The semi-continuity assumption may seem restrictive, since it rules out the discrete distributions. However, the disturbances in many realistic systems can satisfy the semi-continuity because realistic noises are usually generated from a mixture distribution where at least one component is continuous, e.g. the wind disturbances and renewable generations are continuous.

As for the control inputs $u_t$, we first impose the same assumption as Assumption \ref{ass: bounded iid semi continuous noises} for simplicity. Later in Section \ref{sec:results}, we will also discuss the relaxation of this assumption to include control policies.

\begin{assumption}[Bounded i.i.d. and semi-continuous inputs]\label{ass: bounded iid semi continuous inputs}
$u_{t}$ is i.i.d. following a semi-continuous distribution with zero mean and a positive definite covariance $\Sigma_u \succeq \sigma_{u}^{2}I_{n_{x}}\succ 0$ and bounded support, i.e. $\|u_t\|_\infty \leq u_{\max}$ almost surely for all $t$.
\end{assumption}

Lastly, we introduce our stability assumption based on the  input-to-state stability definition below.

\begin{definition}[Locally input-to-state stability (LISS)]\label{def:liss}  
Consider the general nonlinear system $x_{t+1} = f(x_{t}, d_{t})$ with $x_{t} \in \mathbb{R}^{n_{x}}$, $d_{t}\in\mathbb{R}^{n_{d}}$, $f$ being a continuous function such that $f(0, 0) = 0$. This system is called locally input-to-state stable (LISS) if there exist constants $\rho_{x} > 0$, $\rho > 0$ and functions $\gamma \in \mathcal{K}$, $\beta \in \mathcal{KL}$  such that for all $x_{0} \in \{ x_{0} \in \mathbb{R}^{n_{x}}: \|x_{0}\|_{2} \leq \rho_{x} \}$ and any input $d_{t} \in \{d \in \mathbb{R}^{n_{d}} : \sup_{t}\|d_{t}\|_{\infty} \leq \rho\}$ , it holds that $\|x_{t}\|_{2} \leq \beta\big(\|x_{0}\|_{2},t\big) + \gamma\big(\sup_{t}\|d_{t}\|_{\infty}\big)$ for all $t \geq 0$.\footnote{A function $\gamma:\mathbb{R}_{\geq 0} \rightarrow \mathbb{R}_{\geq 0}$ is a $\mathcal{K}$ function if it is continuous, strictly increasing and $\gamma(0) = 0$. 
A function $\beta:\mathbb{R}_{\geq 0} \times \mathbb{R}_{\geq 0} \rightarrow \mathbb{R}_{\geq 0}$ is a $\mathcal{KL}$ function if, for each fixed $t \geq 0$, the function $\beta(\cdot, t)$ is a $\mathcal{K}$ function, and for each fixed $s \geq 0$, the function $\beta(s,\cdot)$ is decreasing and $\beta(s,t) \rightarrow 0$  as $t \rightarrow \infty$.}
\end{definition}

\begin{assumption}[LISS system]\label{ass: iss system}
System \eqref{eq:sys} is LISS with parameters $\rho_x$ and $\rho$ such that $\rho_x\geq \|x_0\|_{2}$ and $\rho\geq \max(w_{\max}, u_{\max})$, respectively.    
\end{assumption}

Assumption \ref{ass: iss system} is imposed, together with the bounded disturbances and inputs in Assumptions \ref{ass: bounded iid semi continuous noises} and \ref{ass: bounded iid semi continuous inputs}, to guarantee bounded states during the control dynamics (for instance, see the proof of Theorem~\ref{th:BMSB-o-l}in Appendix~\ref{appendix:A}). Notably, many studies on  learning-based nonlinear control  require certain  boundedness on the states  for theoretical analysis~\cite{sattar2022non,foster2020learning,li2023online}. 

In addition, it is interesting to note that this paper only requires local stability of the dynamics, whereas several learning-based nonlinear control papers assume certain global properties, such as global exponential stability in~\citep{foster2020learning}, global exponential incremental stability in~ \citep{sattar2022non,li2023online,lin2024online}, or global Lipschitz smoothness in \citep{lee2024active}.\footnote{Global Lipschitz smoothness may exclude system dynamics with higher-order polynomials.} This difference in the dynamics assumption reflects a \textit{trade-off} with the disturbance assumptions: we assume a stronger assumption on the boundedness of disturbances and a weaker assumption on local stability, whereas much of the literature considers (sub)Gaussian distributions (which can be unbounded) but requires stronger global properties for dynamics. Since this paper is largely motivated by physical systems, which typically encounter bounded disturbances/inputs and generally only satisfy local stability~\citep{slotine1991applied}, we address this trade-off through our current set of assumptions, leaving it as an exciting future direction to consider relaxing these assumptions.
\section{Main Results}\label{sec:results}
In this section, we provide the estimation error bounds of LSE for linearly parameterized nonlinear systems under i.i.d. random inputs. The estimation error bounds rely on the establishment of probabilistic persistent excitation, which will be introduced in the first subsection. Later,
we also generalize the results to include control policies and discuss the convergence rate of another popular uncertainty quantification method in the control literature, set membership estimation, whose formal definition is deferred to the corresponding subsection.

\subsection{Probabilistic Persistent Excitation}
{It is well-known that persistent excitation (PE) is a crucial condition for successful system identification~\citep{narendra1987persistent}. In the following, we  introduce the persistent excitation condition for our linearly parameterized nonlinear systems. }

{\begin{definition}[Persistent excitation~\citep{skantze2000adaptive,sastry2011adaptive}]\label{def:p-e} System \eqref{eq:sys} is  persistently excited if there exist $s > 0$ and $m \geq 1$ such that for any $t_{0} \geq 0$, we have

\begin{align*}
    \dfrac{1}{m} \sum_{t=t_{0}}^{t_{0}+m-1}  
    \phi \big(x_{t}, u_{t}\big) \phi^{\intercal} \big(x_{t}, u_{t}\big) \succeq s^{2} I_{n_{\phi}}.
\end{align*}    
\end{definition}}

{In the stochastic setting, PE is closely related with a block-martingale small-ball (BMSB) condition  proposed in \cite{simchowitz2018learning}, which can be viewed as a probabilistic version of PE.}

\begin{definition}[BMSB~\citep{simchowitz2018learning}]\label{def:bmsb}
Let $\{\mathcal{F}_{t}\}_{t \geq 1}$ denote a filtration and let $\{y_{t}\}_{t \geq 1}$ be an $\{\mathcal{F}_{t}\}_{t \geq 1}$-adapted random process taking values in $\mathbb{R}^{n_{y}}$. We say $\{y_{t}\}_{t \geq 1}$ satisfies the $(k, \Gamma_{sb}, p)$-block martingale small-ball (BMSB) condition for a positive integer $k$, a  $\Gamma_{sb} \succ 0$, and a $p \in [0, 1]$, if  for any fixed $v \in \mathbb{R}^{n_{y}}$ such that $\|v\|_{2} = 1$, the process $\{y_{t}\}_{t \geq 1}$ satisfies $\frac{1}{k} \sum_{i=1}^{k} \mathbb{P}\big( |v^{\intercal} y_{t+i}| \geq \sqrt{v^{\intercal} \Gamma_{sb} v}\ |\ \mathcal{F}_{t}\big) \geq p$ almost surely for any $t \geq 1$.
\end{definition}

{One major contribution of this paper is formally establishing the BMSB condition  for linearly parameterized nonlinear systems with real-analytic feature functions.}

{In the following, we first investigate the open-loop system with i.i.d. inputs  and later extend the results to the closed-loop systems with inputs $u_{t} = \pi(x_{t}) + \eta_{t}$, where $\eta_{t}$ represents the noise and $\pi:\mathbb{R}^{n_x} \rightarrow \mathbb{R}^{n_u}$ denotes a control policy. The following theorem considers the open-loop systems.}

\begin{theorem}[{BMSB for open-loop systems}]\label{th:BMSB-o-l} 
Let $u_{t} = \eta_{t}$ and consider the filtration $\mathcal{F}_{t} = \mathcal{F}(w_{0}, \cdots, w_{t-1}, x_{0}, \cdots, x_{t}, \eta_{0}, \cdots, \eta_{t})$. Suppose Assumptions~\ref{ass: analytic},~\ref{ass: bounded iid semi continuous noises},~\ref{ass: bounded iid semi continuous inputs},~\ref{ass: iss system} hold, then there exist $s_{\phi} > 0$ and $p_{\phi} \in (0, 1)$ such that the $\{\mathcal{F}_{t}\}_{t \geq 1}$-adapted process $\big\{\phi\big(x_{t}, u_{t}\big)\big\}_{t \geq 1}$ satisfies the
$\big(1, s_{\phi}^{2}I_{n_{\phi}}, p_{\phi}\big)$-BMSB condition. 
\end{theorem}

\textit{Proof Sketch.}
{Intuitively, BMSB requires that any linear combination of feature functions remains positive with a non-vanishing probability. Notice that a linear combination of real-analytic functions is itself real-analytic, and the zeros of an analytic function have measure zero. These facts allow us to show that the probability of a linear combination of linearly independent feature functions equaling zero is less than one, as long as the noises follow semi-continuous distributions, by the connection of the Lebesgue measure  and the probability measure in Definition~\ref{def: semi continuous}.} 

In more detail, the proof leverages a variant of the Paley-Zygmund argument~\citep{petrov2007lower}, which provides a lower bound for the tail properties of positive random variables. Specifically, it states that the probability of a positive random variable being small depends on the ratio of its even moments. We apply this result to the random variable $|v^{T} \phi\big(x_{t+1}, u_{t+1}\big) \ |\ \mathcal{F}_{t}|$ with $\|v\|_{2} = 1$ and aim to show that the lower bound is non-trivial for any direction $v$ with $\|v\|_{2} = 1$ and any filtration $\mathcal{F}_{t}$, $t \geq 0$. We then use results from measure theory to demonstrate the existence of such a non-trivial lower bound. This is done by showing that the Lebesgue measure of the set where $|v^{T} \phi\big(x_{t+1}, u_{t+1}\big)| = 0$ is zero, and thus the even moments of $|v^{T} \phi\big(x_{t+1}, u_{t+1}\big) \ |\ \mathcal{F}_{t}|$ are non-zero, provided that the noise and disturbance distributions are semi-continuous. For further details, please refer to Appendix~\ref{appendix:A}.\qed

{It is worth pointing out that Theorem \ref{th:BMSB-o-l}  only establishes the existence of the constants $(s_{\phi}, p_{\phi})$, and deriving explicit formulas of these constants are left for future work. In particular, it can be challenging to derive a generic formula for all linearly parameterized nonlinear systems, but an exciting direction is to study reasonable sub-classes of systems and construct their corresponding formulas of the constants $(s_{\phi}, p_{\phi})$.}

\subsection{Non-asymptotic Bounds for LSE}
We are now prepared to present the non-asymptotic convergence rate for the LSE in learning the unknown parameters of the system~\eqref{eq:sys}.

\begin{theorem}[{LSE's convergence rate for open-loop systems}]\label{th:lse-con-o-l}
Consider the dynamical system described in~(\ref{eq:sys}) with i.i.d.  inputs $u_t=\eta_t$ and assume that Assumptions~\ref{ass: analytic}, \ref{ass: bounded iid semi continuous noises}, \ref{ass: bounded iid semi continuous inputs}, \ref{ass: iss system} are satisfied. Let $s_{\phi}$ and $p_{\phi}$ be as defined in Theorem~\ref{th:BMSB-o-l}, and define $\bar{b}_{\phi} = \sup_{t \geq 0} \mathbb{E}\big[\|\phi(z_{t})\|_{2}^{2}\big]$. For a fixed $\delta \in (0, 1)$ and $T \geq 1$, if $T$ satisfies the condition
\begin{align*}
    T \geq \dfrac{10}{p_{\phi}} \bigg( \log\bigg(\dfrac{1}{\delta}\bigg) + 2n_{\phi}\log\bigg(\dfrac{10}{p_{\phi}}\bigg) +  n_{\phi}\log\bigg(\dfrac{\bar{b}_{\phi}}{\delta s_{\phi}^{2}}\bigg) \bigg),
\end{align*}
then LSE's estimation $\hat \theta_T$ satisfies the following error bound with probability at least $1-3\delta$.% the error of LSE (in~(\ref{eq:ols-eq})) in learning the unknown parameters from random sequence $\{(x_{t}, u_{t}, x_{t+1})\}_{t \geq 1}^{T}$ satisfying ~(\ref{eq:sys}) is bounded as follows:
\begin{align*}
 \big \|\hat{\theta}_{T} - \theta_{*} \big\|_{2} \leq \dfrac{90 \sigma_{w}}{p_{\phi}} \sqrt{ \dfrac{n_{x} +\log\bigg(\frac{1}{\delta}\bigg) + n_{\phi}\log\bigg(\frac{10}{p_{\phi}}\bigg) +  n_{\phi}\log\bigg(\frac{\bar{b}_{\phi}}{\delta s_{\phi}^{2}}\bigg)}{Ts_{\phi}^{2}}}.
\end{align*}
%with probability at least $1-3\delta$.
\end{theorem}

{The proof relies on Theorem \ref{th:BMSB-o-l} and Theorem 2.4 in~\citep{simchowitz2018learning}. The complete proof is provided in Appendix~\ref{app:B1}.}%, combined with the probabilistic persistent excitation of features (Theorem~\ref{th:BMSB-o-l}) and the fact that $\bar b_{\phi}$ is bounded, given that the features are analytic, the system is LISS, and the noises and disturbances have bounded support. Additional details is provided in Appendix~\ref{app:B1}.}

Theorem \ref{th:lse-con-o-l} demonstrates that LSE converges to the true parameters under random control inputs and random disturbances (non-active exploration) at a rate of $\frac{1}{\sqrt T}$  for linearly parameterized nonlinear systems. This is consistent with the convergence rates of LSE for linear systems in terms of $T$. 

{Regarding the dimension dependence in the convergence rate,  the explicit dependence is  $\sqrt{n_x+n_\phi}$, where $n_x$ and $n_\phi$ refer to the dimensions of the state and the feature vector, respectively. Besides, it is worth mentioning that other parameters, such as $s_{\phi}, p_{\phi}, \bar b_{\phi}$, may implicitly depend on the dimensions as well. For some special systems, such as bilinear systems, it has been shown that these constants are independent of the dimensions~\citep{sattar2022finite}. It is left as future work to explore other nonlinear systems' implicit dimension dependence.

Next, we can  generalize the i.i.d. inputs $u_t$ to include  control policies, i.e., $u_t=\pi(x_t)+\eta_t$, where $\eta_t$ satisfies Assumption \ref{ass: bounded iid semi continuous inputs} and $\pi(x_t)$ is analytic.

\begin{corollary}[{LSE's convergence rate for closed-loop systems}]\label{cor:lse-con-c-l}
Consider inputs $u_t=\pi(x_t)+\eta_t$, where {$\pi(\cdot)$ is real-analytic}, $\eta_t$ satisfies Assumption \ref{ass: bounded iid semi continuous inputs}, and the closed-loop system $x_{t+1}=\theta_*\phi(x_t, \pi(x_t)+\eta_t)+w_t$ satisfies Assumption \ref{ass: iss system} for both $w_t$ and $\eta_t$. Then, the same convergence rate in Theorem \ref{th:lse-con-o-l}  holds.
\end{corollary}
The proof is provided in Appendix~\ref{app:B2}.

\subsection{Non-asymptotic Diameter Bounds for SME}
Set membership estimation (SME) is another popular method for uncertainty quantification in control system estimation \citep{lu2023robust,bertsekas1971control,li2024icml}. Unlike LSE, SME is a set-estimator and directly estimates the uncertainty set.  Since the analysis of SME also relies on the probabilistic persistent excitation analysis, we can also establish the convergence rate of SME for linearly parameterized nonlinear systems under i.i.d. noises in the following. In particular, SME estimates the uncertainty set as
\begin{align}\label{eq:SM-eq}
    \Theta_{T} = \bigcap \limits_{t=0}^{T-1} \bigg\{\hat{\theta}: x_{t+1} - \hat{\theta} \phi(x_{t}, u_{t}) \in \mathcal{W} \bigg\},
\end{align}
where $\mathcal{W}$ is a bounded set such that $w_{t} \in \mathcal{W}$ for all $t \geq 0$.

The convergence of SME relies on an additional assumption as shown below: the tightness of the bound $\mathcal W$ on $w_t$'s support. This tightness assumption is commonly considered in SME's literature \citep{li2024icml,lu2019robust,akccay2004size}. Further, \citep{li2024icml}  discusses the relaxation of this assumption by learning a tight bound at the same time of learning the uncertainty set of $\theta_*$ for linear systems. Similar tricks can be applied to nonlinear systems, but this paper only considers the vanilla case of SME for simplicity.

\begin{assumption}[Tight bound on disturbances]\label{ass:w-tight}
Assume for any $\epsilon > 0$, there exists $q_{w}(\epsilon) > 0$, such that for any $1 \leq j \leq n$ and $t \geq 0$, we have $\mathbb{P}(w_{t}^{j} + w_{\max} \leq \epsilon) \geq q_{w}(\epsilon) > 0$, $\mathbb{P}(w_{\max} - w_{t}^{j} \leq \epsilon) \geq q_{w}(\epsilon) > 0$.    
\end{assumption}
Assumption~\ref{ass:w-tight} requires that $w_t$ can visit set $\mathcal W$'s boundary arbitrarily closely with a positive probability. For example, for a one-dimensional $w_{t}$ bounded by $-w_{\max} \leq w_{t} \leq w_{\max}$, Assumption~\ref{ass:w-tight} requires that there is a positive probability that $w_{t}$ is close to $w_{\max}$ and $-w_{\max}$, i.e., for any $\epsilon > 0$, we have $\mathbb{P}( w_{\max} -\epsilon \leq w_{t} \leq w_{\max}) > 0$ and $\mathbb{P}( -w_{\max} \leq w_{t} \leq - w_{\max} + \epsilon) > 0$.

Next, we state a non-asymptotic bound on the diameter of the uncertainty set estimated by SME.

\begin{theorem}[SME's diameter bound for open-loop systems]\label{th:sme-con-o-l}
Consider  system (\ref{eq:sys}) with i.i.d.  inputs $u_t=\eta_t$. Suppose  Assumptions~\ref{ass: analytic}, \ref{ass: bounded iid semi continuous noises}, \ref{ass: bounded iid semi continuous inputs}, \ref{ass: iss system} are satisfied. Consider $s_{\phi}$ and $p_{\phi}$  defined in Theorem~\ref{th:BMSB-o-l} and let $b_{\phi} = \sup_{t \geq 0} \|\phi(z_{t})\|_{2}$.  For any $m \geq 0$ and  $\delta \in (0, 1)$, when $T > m$, we have
\begin{align*}
    \mathbb{P}\bigg(\textup{diam}(\Theta_{T}) > \delta\bigg) \leq \frac{T}{m} \tilde{O}\big(n_{\phi}^{2.5}\big)a_{2}^{n_{\phi}}\exp(-a_{3}m)\ +\  \tilde{O}\big(n_{x}^{2.5}n_{\phi}^{2.5}\big)a_{4}^{n_{x}n_{\phi}}\bigg(1-q_{w} \bigg(\frac{a_{1}\delta}{4\sqrt{n_{x}}}\bigg)\bigg)^{\frac{T}{m}},
\end{align*}
where $a_{1} = \frac{s_{\phi}p_{\phi}}{4}$, $a_{2} = \frac{64b_{\phi}^{2}}{s_{\phi}^{2}p_{\phi}^{2}}$, $a_{3} = \frac{p_{\phi}^{2}}{8}$, 
$a_{4} = \frac{16b_{\phi}\sqrt{n_{x}}}{s_{\phi}p_{\phi}}$. 
The constants hidden in $\tilde{O}$ are provided in the Appendix~\ref{app:C1}.
\end{theorem}

The proof of Theorem \ref{th:sme-con-o-l} relies on Theorem~\ref{th:BMSB-o-l} in this paper and Theorem 1 from~\citep{li2024icml}. The detailed proof is provided in Appendix~\ref{app:C1}. 

Theorem \ref{th:sme-con-o-l} establishes an upper bound on the "failure" probability of SME, i.e., the probability that the uncertainty set's diameter exceeds $\delta$. To ensure the failure probability is less than $1$, one can select $m=O(\log(T))$ and choose a sufficiently large $T$ such that $T \geq m=O(\log(T))$. If $w_{t}$ is more likely to visit the boundaries of the set $\mathcal{W}$ (meaning a larger $q(\ell)$), SME is less likely to estimate an uncertainty set with a diameter greater than $\delta$. 

To provide more intuitions on the diameter bound in Theorem~\ref{th:sme-con-o-l}, we consider $q_{w}(\ell) = c_{w} \ell$ for some $c_{w} > 0$. Note that several common distributions, including the uniform distribution and the truncated Gaussian distribution, satisfy this property on $q_{w}(\ell)$ (see Appendix~\ref{app:C2} for explicit formulas of $c_w$). With $q_{w}(\ell) = c_{w} \ell$, we can provide a convergence rate of SME in terms of $T$ in the following.

{\begin{corollary}[SME's convergence rate when $q_{w}(\ell) = c_{w} \ell$
]\label{cor:sme-rate}
For any $\epsilon > 0$, let 
$$m \geq O\Bigg( \dfrac{\log\big(\frac{T}{\epsilon}\big) + n_{\phi} \log\big(\frac{8 b_{\phi}}{s_{\phi}p_{\phi}}\big)}{p_{\phi}^{2}} \Bigg).$$ 
If $w_{t}$'s distribution satisfies $q_{w}(\ell) = c_{w} \ell$ for all $\ell > 0$, then with probability at least $1-2\epsilon$, we have:
\begin{align*}
    \emph{diam}(\Theta_{T}) \leq 
    O\Bigg( \dfrac{\sqrt{n_{x}} \log\big(\frac{1}{\epsilon}\big) + n_{x}^{1.5}n_{\phi}\log\big(\frac{b_{\phi}n_{x}}{s_{\phi}p_{\phi}}\big)}{c_{w}s_{\phi}p_{\phi}T} \Bigg),
\end{align*}
where the constants hidden in $O(\cdot)$ are provided in  Appendix C.
\end{corollary}}
The proof of Corollary~\ref{cor:sme-rate} is provided in Appendix~\ref{app:C2}.

Finally, similar to LSE, we can extend  SME's convergence rates from open-loop systems to closed-loop systems with real-analytic control policies, i.e., $u_t=\pi(x_t)+\eta_t$, where $\eta_t$ satisfies Assumption \ref{ass: bounded iid semi continuous inputs} and $\pi(x_t)$ is real-analytic.

\begin{corollary}[SME's convergence rate for closed-loop systems]\label{cor:sme-con-c-l}
Consider inputs $u_t=\pi(x_t)+\eta_t$, where  $\pi(\cdot)$ is real-analytic, $\eta_t$ satisfies Assumption \ref{ass: bounded iid semi continuous inputs}, and the closed-loop system $x_{t+1}=\theta_*\phi(x_t, \pi(x_t)+\eta_t)+w_t$ satisfies Assumption \ref{ass: iss system} in terms of both $w_t$ and $\eta_t$. Then, the same convergence rates in Theorem~\ref{th:sme-con-o-l} and Corollary~\ref{cor:sme-rate} still hold.
\end{corollary}
The proof of Corollary \ref{cor:sme-con-c-l} is provided in Appendix~\ref{app:C3}.

\section{Numerical Experiments}\label{sec:num-exp}

In this section, we evaluate the performance of LSE in estimating the unknown parameter $\theta_{*}$ and SME in estimating the uncertainty set for the unknown parameters using the pendulum and quadrotor examples outlined in Section~\ref{sec:prob}. We compare the empirical convergence rates of LSE and SME with the theoretical rates in Theorem~\ref{th:lse-con-o-l} and Corollary~\ref{cor:sme-rate}. In each case, the input $u_{t}$ is composed of a control policy and i.i.d. noise, such that $u_{t} = \pi(x_{t}) + \eta_{t}$. For our experiments, we employ noise and disturbances drawn from uniform and truncated-Gaussian distributions. To compute theoretical rates, we numerically estimate parameters such as $s_{\phi}$ and $p_{\phi}$ (see Appendix~\ref{appendix:E}). Further details can be found in our source code.\footnote{\footnotesize \url{https://github.com/NeginMusavi/real-analytic-nonlinear-sys-id}} 

The details for these scenarios are outlined below: 
\begin{itemize}
    \item \textbf{Pendulum~\ref{example: pendulum}:} In the pendulum example described in Section~\ref{sec:prob}, the control input is $u_{t} = - k \dot{\alpha}_{t} + \eta_{t}$. This scenario includes two unknown parameters: $\theta_{1} = \dfrac{1}{l}$ and $\theta_{2} = \dfrac{1}{ml^{2}}$.
    \item \textbf{Quadrotor~\ref{example: quadrotor}:} For the quadrotor example in Section~\ref{sec:prob}, the control input is defined as $u_t = \pi(x_t) + \eta_t$, where $\pi(x_t)$ follows the controller proposed by~\citet{alaimo2013mathematical}. The quadrotor system involves $13$ states and $4$ inputs, with the unknown parameter matrix $\theta_{*}$ containing $7$ parameters, including the mass $m$ and specific elements of the inertia matrix $I$.
\end{itemize}
Further details on controller gains and unknown parameters are provided in Appendix~\ref{appendix:D}.

\paragraph{LSE Results:}
Figures~\ref{fig:lse_pend_uni} and~\ref{fig:lse_pend_trunc} present a comparison between the LSE theoretical bound from Theorem~\ref{th:lse-con-o-l} with its empirical estimation error of the unknown parameters $\theta_{*}$ versus trajectory length $T$ for the pendulum example, with uniform and truncated-Gaussian noises and disturbances. Similarly, Figures~\ref{fig:lse_quad_uni} and~\ref{fig:lse_quad_trunc} show this comparison for the quadrotor example. In each figure, both the theoretical bound and empirical error are normalized by the $l_2$ norm of the nominal parameter $\theta_{*}$. The log-log plots for both scenarios demonstrate that the empirical error rate achieves $O(\frac{1}{\sqrt{T}})$ which in consistent with the theoretical rate in Theorem~\ref{th:lse-con-o-l}.
\begin{figure}[ht]
    \centering
    \subcaptionbox{Uniform\hspace*{-0.75cm} \label{fig:lse_pend_uni}}[0.24\textwidth]{      
        \includegraphics[width=\linewidth]{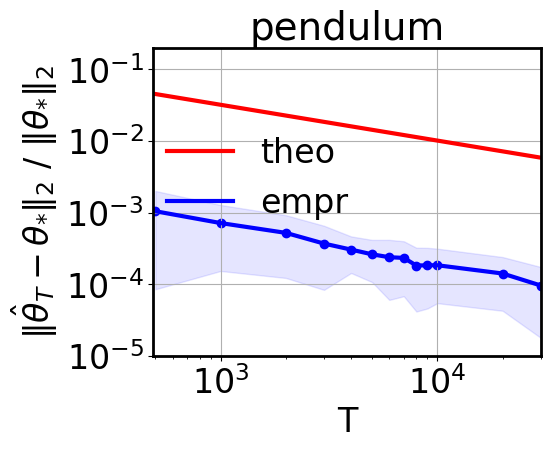}
    }
    % \hfill
    \subcaptionbox{Truncated-Gaussian\hspace*{-0.4cm}\label{fig:lse_pend_trunc}}[0.24\textwidth]{   
        \includegraphics[width=\linewidth]{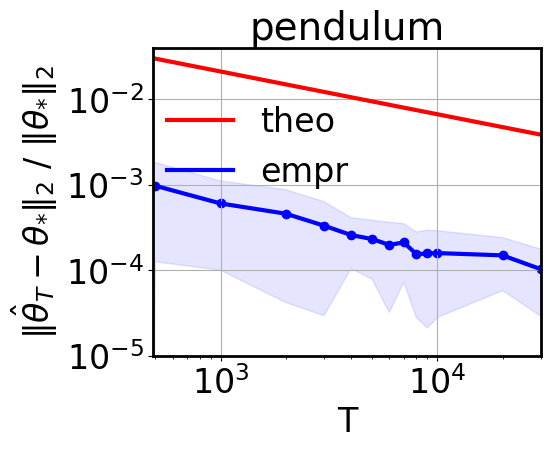}
    }
    \subcaptionbox{Uniform\hspace*{-0.9cm}\label{fig:lse_quad_uni}}[0.24\textwidth]{     
        \includegraphics[width=\linewidth]{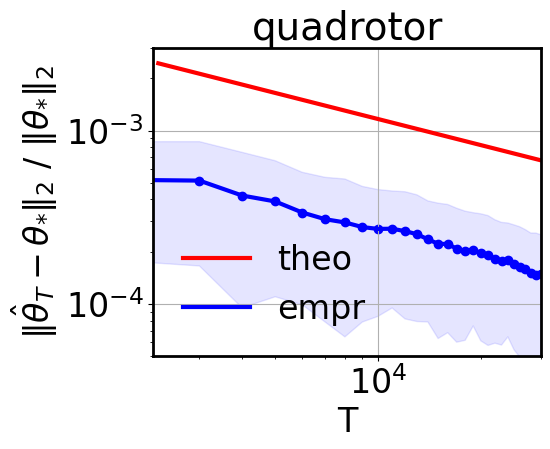}
    }
    \subcaptionbox{Truncated-Gaussian\hspace*{-0.75cm}\label{fig:lse_quad_trunc}}[0.24\textwidth]{
        \includegraphics[width=\linewidth]{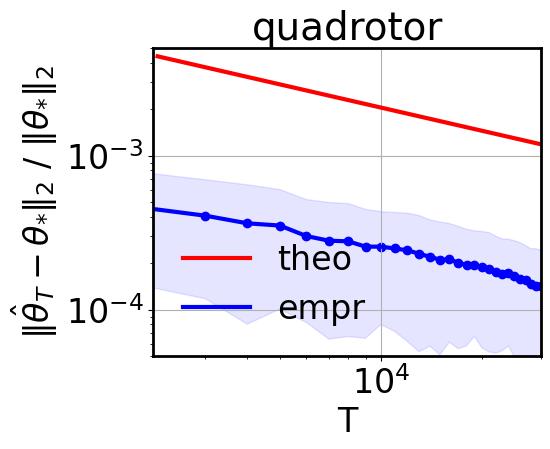}
    }
    \caption{\footnotesize Convergence rate of the LSE for pendulum and quadrotor scenarios: (a) Pendulum example with uniform, (b) Pendulum example with truncated-Gaussian, (c) Quadrotor example with uniform, and (d) Quadrotor example with truncated-Gaussian noises and disturbances.
    Here, uniform noises and disturbances are i.i.d. generated from $\texttt{uniform}([-1, 1])$, and truncated-Gaussian noises and disturbances are i.i.d. generated from $\texttt{truncated-Gaussian}(0, 0.1, [-1, 1])$. "theo" denotes the theoretical convergence rate, and "empr" represents the empirical rate. The mean error across 20 trials is shown by dots on the empirical plots, with shaded areas illustrating empirical standard deviation.}    
\end{figure}

\paragraph{SME Results:} Figures~\ref{fig:sme_pend_uni} and~\ref{fig:sme_pend_trunc} show the empirical convergence rate of SME for the pendulum example, for uniform and truncated-Gaussian noises and disturbances, in comparison to the theoretical rate from Corollary~\ref{cor:sme-rate}. Both the theoretical bound and empirical error in the figures are normalized by the $l_2$ norm of the nominal parameter $\theta_{*}$. The log-log plots indicate that the empirical rate achieves $O(\frac{1}{T})$, which is consistent with the results from Corollary~\ref{cor:sme-rate} and with the related results for linear systems in~\citep{li2024icml}. A similar result can be observed for the quadrotor example, in Figures~\ref{fig:sme_quad_uni} and~\ref{fig:sme_quad_trunc}. Additionally, Figure~\ref{fig:pend_sme} shows the uncertainty sets estimated by SME for the two unknown parameters, labeled \( \theta_{1} \) and \( \theta_{2} \), in the pendulum example, along with the diameters of these sets as trajectory length grows. We observe that these sets contract as trajectory length increases, with the true values of the unknown parameters lying within the estimated uncertainty sets. The illustration of uncertainty sets for the quadrotor example is provided in Appendix~\ref{app:quad}.
\begin{figure}[htp]
    \centering
    \subcaptionbox{Uniform\hspace*{-0.75cm} \label{fig:sme_pend_uni}}[0.24\textwidth]{     
        \includegraphics[width=\linewidth]{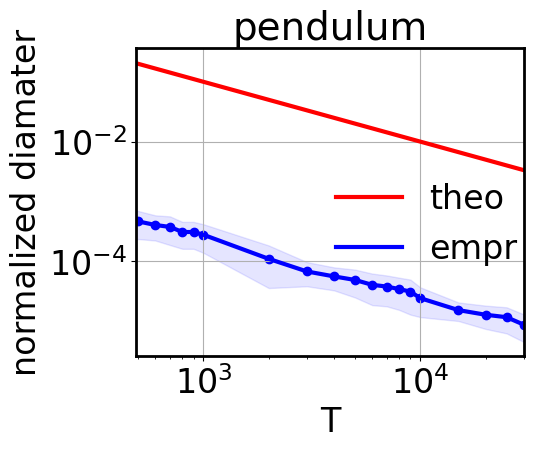}
    }
    % \hfill
    \subcaptionbox{Truncated-Gaussian\hspace*{-0.4cm}\label{fig:sme_pend_trunc}}[0.24\textwidth]{  
        \includegraphics[width=\linewidth]{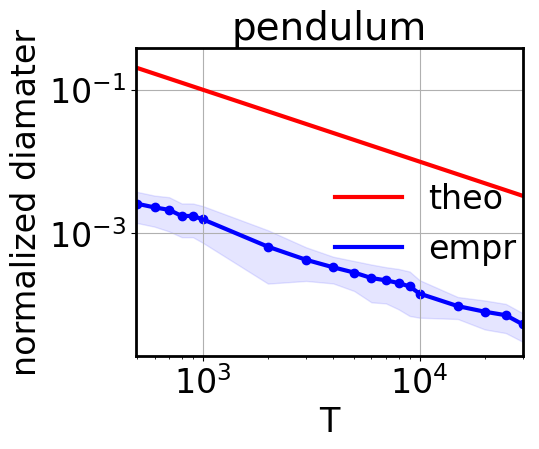}
    }
    \subcaptionbox{Uniform\hspace*{-0.9cm}\label{fig:sme_quad_uni}}[0.24\textwidth]{     
        \includegraphics[width=\linewidth]{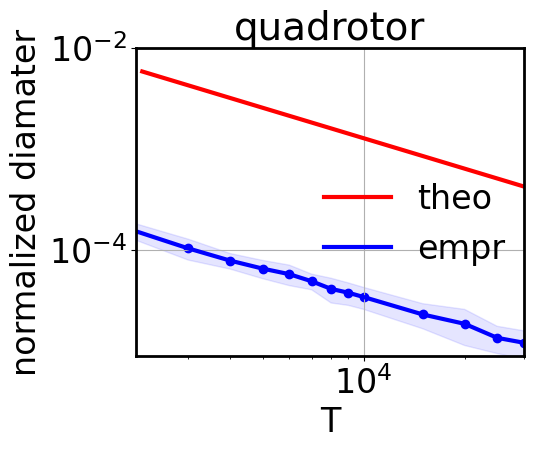}
    }
    \subcaptionbox{Truncated-Gaussian\hspace*{-0.75cm}\label{fig:sme_quad_trunc}}[0.24\textwidth]{  
        \includegraphics[width=\linewidth]{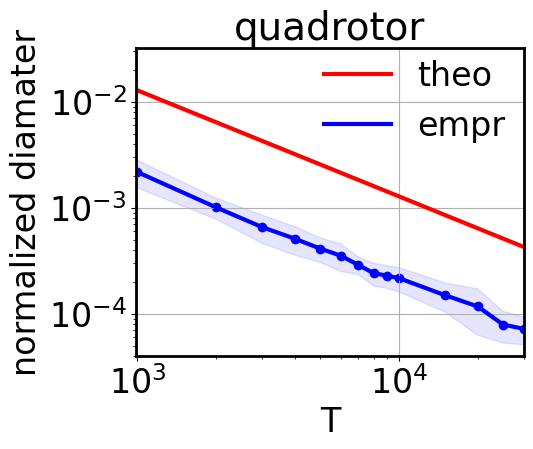}
    }
    \caption{\footnotesize Convergence rate of the SME for pendulum and quadrotor scenarios: (a) Pendulum example with uniform, (b) Pendulum example with truncated-Gaussian, (c) Quadrotor example with uniform, and (d) Quadrotor example with truncated-Gaussian noises and disturbances. Here, uniform noises and disturbances are i.i.d. generated from $\texttt{uniform}([-1, 1])$, and truncated-Gaussian noises and disturbances are i.i.d. generated from $\texttt{truncated-Gaussian}(0, 0.5, [-1, 1])$. "theo" denotes the theoretical convergence rate, and "empr" represents the empirical rate. The mean error across 10 trials is shown by dots on the empirical plots, with shaded areas illustrating empirical standard deviation.}
\end{figure}

\begin{figure}[ht]
    \centering
    \subcaptionbox{Pendulum\hspace{1.8cm}(b) Uncertainty set diameter\hspace{1.6cm}(c) Uncertainty set\hspace*{-0.2cm}}[0.9\textwidth]{      
        \includegraphics[width=\linewidth]{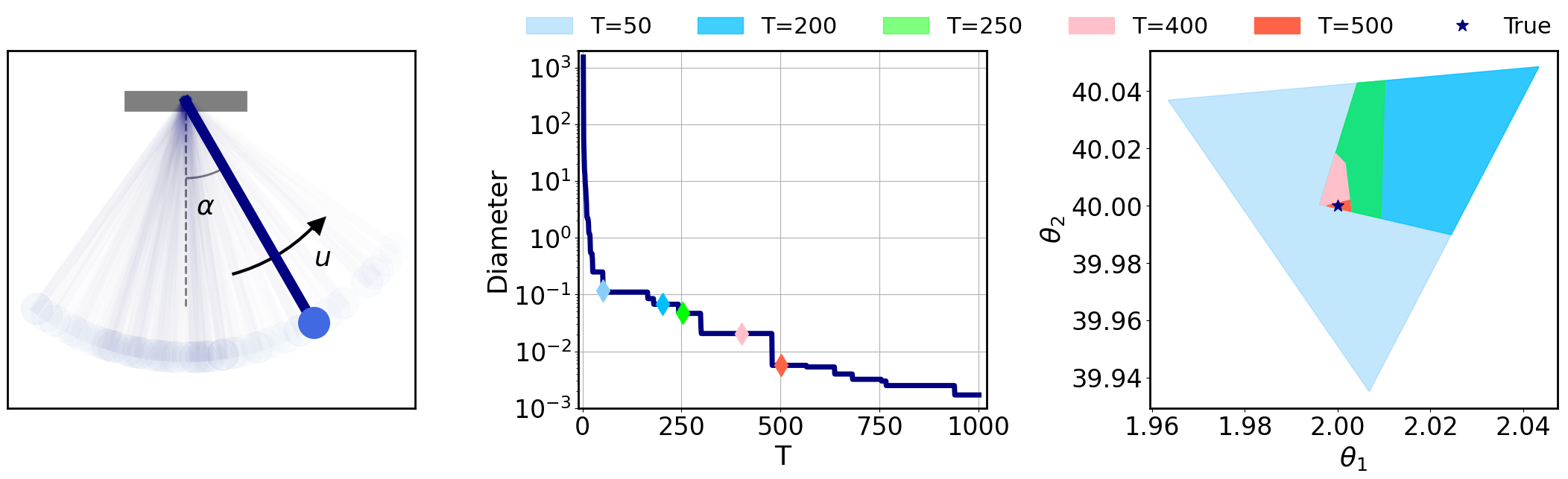}
    }
    \caption{
    \footnotesize Performance of SME for pendulum in (a) with control input $u_{t} = -k\dot{\alpha}_{t}+\eta_{t}$ where $k=0.1$, $\eta_{t}$ i.i.d. generated from $\texttt{truncated-Gaussian}(0, 2, [-2, 2])$ and disturbed with $w_{t}$ i.i.d. generated from $\texttt{truncated-Gaussian}(0, 1, [-1, 1])$. (b) Diameter of the uncertainty set estimated by SME. (c) Uncertainty set depicted for $T=50, 200, 250, 400, 500$.}
    \label{fig:pend_sme}
\end{figure}

\section{Concluding Remarks}\label{sec:con}

\textbf{Conclusion.} This study examines the probabilistic persistent excitation in a class of nonlinear systems influenced by i.i.d. noise and stochastic disturbances, with the stipulation that their distributions do not concentrate on sets of Lebesgue measure zero. Based on this we then present an explicit bound on the convergence rate of SME estimations and LSE estimations for this class of dynamical systems. Additionally, numerical experiments in the context of robotics are provided to illustrate both methods.

\textbf{Limitations.} One limitation of this work is that our analysis relies on a specific class of i.i.d. noises and stochastic disturbances, where the probability distribution is not concentrated on sets of Lebesgue measure zero. While this is a sufficient condition, it is possible that the BMSB conditions are satisfied under other circumstances. Another limitation is that, though we provide sufficient conditions for the existence of parameters satisfying the BMSB condition, the explicit dependence is not detailed here. Lastly, imperfect observations are not considered here. 

\textbf{Future Work.} {Our future work includes several promising directions, e.g., to explore cases that do not satisfy our semi-continuity assumption, such as discrete noises, and to investigate the explicit dependence of the BMSB parameter on system attributes, such as state, input, and feature dimensions, etc. Furthermore, extending this work to imperfect state observations is an important next step. Finally, a potential direction is to provide a non-asymptotic analysis of the volumes of uncertainty sets estimated by SME uncertainty sets, as opposed to the current focus on their diameters.}

\section*{Broader Impact}
This paper is a foundation research and develops theoretical insight to estimation of nonlinear control systems. We do not see a direct path to negative applications in general. But we want to mention that successful applications of our theoretical results rely on verifying the assumptions in this paper.

\bibliographystyle{plainnat.bst}
\bibliography{references}

\begin{thebibliography}{51}
\providecommand{\natexlab}[1]{#1}
\providecommand{\url}[1]{\texttt{#1}}
\expandafter\ifx\csname urlstyle\endcsname\relax
  \providecommand{\doi}[1]{doi: #1}\else
  \providecommand{\doi}{doi: \begingroup \urlstyle{rm}\Url}\fi

\bibitem[Abbasi-Yadkori and Szepesv{\'a}ri(2011)]{abbasi2011regret}
Yasin Abbasi-Yadkori and Csaba Szepesv{\'a}ri.
\newblock Regret bounds for the adaptive control of linear quadratic systems.
\newblock In \emph{Proceedings of the 24th Annual Conference on Learning Theory}, pages 1--26. JMLR Workshop and Conference Proceedings, 2011.

\bibitem[Ak{\c{c}}ay(2004)]{akccay2004size}
H{\"u}seyin Ak{\c{c}}ay.
\newblock The size of the membership-set in a probabilistic framework.
\newblock \emph{Automatica}, 40\penalty0 (2):\penalty0 253--260, 2004.

\bibitem[Alaimo et~al.(2013)Alaimo, Artale, Milazzo, Ricciardello, and Trefiletti]{alaimo2013mathematical}
Andrea Alaimo, Valeria Artale, C~Milazzo, Angela Ricciardello, and LUCA Trefiletti.
\newblock Mathematical modeling and control of a hexacopter.
\newblock In \emph{2013 International conference on unmanned aircraft systems (ICUAS)}, pages 1043--1050. IEEE, 2013.

\bibitem[Bai et~al.(1995)Bai, Tempo, and Cho]{bai1995membership}
Er-Wei Bai, Roberto Tempo, and Hyonyong Cho.
\newblock Membership set estimators: size, optimal inputs, complexity and relations with least squares.
\newblock \emph{IEEE Transactions on Circuits and Systems I: Fundamental Theory and Applications}, 42\penalty0 (5):\penalty0 266--277, 1995.

\bibitem[Bai et~al.(1998)Bai, Cho, and Tempo]{bai1998convergence}
Er-Wei Bai, Hyonyong Cho, and Roberto Tempo.
\newblock Convergence properties of the membership set.
\newblock \emph{Automatica}, 34\penalty0 (10):\penalty0 1245--1249, 1998.

\bibitem[Bertsekas(1971)]{bertsekas1971control}
Dimitri~P Bertsekas.
\newblock \emph{Control of uncertain systems with a set-membership description of the uncertainty.}
\newblock PhD thesis, Massachusetts Institute of Technology, 1971.

\bibitem[Bogachev and Ruas(2007)]{bogachev2007measure}
Vladimir~Igorevich Bogachev and Maria Aparecida~Soares Ruas.
\newblock \emph{Measure theory}, volume~1.
\newblock Springer, 2007.

\bibitem[Chen and Hazan(2021)]{chen2021black}
Xinyi Chen and Elad Hazan.
\newblock Black-box control for linear dynamical systems.
\newblock In \emph{Conference on Learning Theory}, pages 1114--1143. PMLR, 2021.

\bibitem[Cr{\u{a}}ciun and Ghoshdastidar(2024)]{cruaciun2024stability}
Alexandru Cr{\u{a}}ciun and Debarghya Ghoshdastidar.
\newblock On the stability of gradient descent for large learning rate.
\newblock \emph{arXiv preprint arXiv:2402.13108}, 2024.

\bibitem[Dean et~al.(2018)Dean, Mania, Matni, Recht, and Tu]{dean2018regret}
Sarah Dean, Horia Mania, Nikolai Matni, Benjamin Recht, and Stephen Tu.
\newblock Regret bounds for robust adaptive control of the linear quadratic regulator.
\newblock \emph{Advances in Neural Information Processing Systems}, 31:\penalty0 4188--4197, 2018.

\bibitem[Foster et~al.(2020)Foster, Sarkar, and Rakhlin]{foster2020learning}
Dylan Foster, Tuhin Sarkar, and Alexander Rakhlin.
\newblock Learning nonlinear dynamical systems from a single trajectory.
\newblock In \emph{Learning for Dynamics and Control}, pages 851--861. PMLR, 2020.

\bibitem[Gao et~al.(2024)Gao, Tang, Qi, and Yang]{gao2024closure}
Yihuai Gao, Yukai Tang, Han Qi, and Heng Yang.
\newblock Closure: Fast quantification of pose uncertainty sets.
\newblock \emph{arXiv preprint arXiv:2403.09990}, 2024.

\bibitem[Halmos(2013)]{halmos2013measure}
Paul~R Halmos.
\newblock \emph{Measure theory}, volume~18.
\newblock Springer, 2013.

\bibitem[Khosravi(2023)]{khosravi2023representer}
Mohammad Khosravi.
\newblock Representer theorem for learning koopman operators.
\newblock \emph{IEEE Transactions on Automatic Control}, 2023.

\bibitem[Kim and Lavaei(2024)]{kimonline}
Jihun Kim and Javad Lavaei.
\newblock Online bandit control with dynamic batch length and adaptive learning rate.
\newblock 2024.

\bibitem[Kong et~al.(2015)Kong, Pfeiffer, Schildbach, and Borrelli]{kong2015kinematic}
Jason Kong, Mark Pfeiffer, Georg Schildbach, and Francesco Borrelli.
\newblock Kinematic and dynamic vehicle models for autonomous driving control design.
\newblock In \emph{2015 IEEE intelligent vehicles symposium (IV)}, pages 1094--1099. IEEE, 2015.

\bibitem[Kowshik et~al.(2021)Kowshik, Nagaraj, Jain, and Netrapalli]{kowshik2021near}
Suhas Kowshik, Dheeraj Nagaraj, Prateek Jain, and Praneeth Netrapalli.
\newblock Near-optimal offline and streaming algorithms for learning non-linear dynamical systems.
\newblock \emph{Advances in Neural Information Processing Systems}, 34:\penalty0 8518--8531, 2021.

\bibitem[Krantz and Parks(2002)]{krantz2002primer}
Steven~G Krantz and Harold~R Parks.
\newblock \emph{A primer of real analytic functions}.
\newblock Springer Science \& Business Media, 2002.

\bibitem[Lee et~al.(2024)Lee, Ziemann, Pappas, and Matni]{lee2024active}
Bruce~D Lee, Ingvar Ziemann, George~J Pappas, and Nikolai Matni.
\newblock Active learning for control-oriented identification of nonlinear systems.
\newblock \emph{arXiv preprint arXiv:2404.09030}, 2024.

\bibitem[Li et~al.(2021{\natexlab{a}})Li, Das, and Li]{li2021online}
Yingying Li, Subhro Das, and Na~Li.
\newblock Online optimal control with affine constraints.
\newblock In \emph{Proceedings of the AAAI Conference on Artificial Intelligence}, volume~35, pages 8527--8537, 2021{\natexlab{a}}.

\bibitem[Li et~al.(2021{\natexlab{b}})Li, Tang, Zhang, and Li]{li2021distributed}
Yingying Li, Yujie Tang, Runyu Zhang, and Na~Li.
\newblock Distributed reinforcement learning for decentralized linear quadratic control: A derivative-free policy optimization approach.
\newblock \emph{IEEE Transactions on Automatic Control}, 67\penalty0 (12):\penalty0 6429--6444, 2021{\natexlab{b}}.

\bibitem[Li et~al.(2023{\natexlab{a}})Li, Preiss, Li, Lin, Wierman, and Shamma]{li2023online}
Yingying Li, James~A Preiss, Na~Li, Yiheng Lin, Adam Wierman, and Jeff~S Shamma.
\newblock Online switching control with stability and regret guarantees.
\newblock In \emph{Learning for Dynamics and Control Conference}, pages 1138--1151. PMLR, 2023{\natexlab{a}}.

\bibitem[Li et~al.(2023{\natexlab{b}})Li, Zhang, Das, Shamma, and Li]{li2023non}
Yingying Li, Tianpeng Zhang, Subhro Das, Jeff Shamma, and Na~Li.
\newblock Non-asymptotic system identification for linear systems with nonlinear policies.
\newblock \emph{IFAC-PapersOnLine}, 56\penalty0 (2):\penalty0 1672--1679, 2023{\natexlab{b}}.

\bibitem[Li et~al.(2024)Li, Yu, Conger, Kargin, and Wierman]{li2024icml}
Yingying Li, Jing Yu, Lauren Conger, Taylan Kargin, and Adam Wierman.
\newblock Learning the uncertainty sets of linear control systems via set membership: A non-asymptotic analysis.
\newblock In \emph{Proceedings of the 41st International Conference on Machine Learning}, pages 29234--29265. PMLR, 2024.
\newblock URL \url{https://proceedings.mlr.press/v235/li24ci.html}.

\bibitem[Lin et~al.(2024)Lin, Preiss, Anand, Li, Yue, and Wierman]{lin2024online}
Yiheng Lin, James~A Preiss, Emile Anand, Yingying Li, Yisong Yue, and Adam Wierman.
\newblock Online adaptive policy selection in time-varying systems: No-regret via contractive perturbations.
\newblock \emph{Advances in Neural Information Processing Systems}, 36, 2024.

\bibitem[Lorenzen et~al.(2019)Lorenzen, Cannon, and Allg{\"o}wer]{lorenzen2019robust}
Matthias Lorenzen, Mark Cannon, and Frank Allg{\"o}wer.
\newblock Robust mpc with recursive model update.
\newblock \emph{Automatica}, 103:\penalty0 461--471, 2019.

\bibitem[Lu and Cannon(2023)]{lu2023robust}
Xiaonan Lu and Mark Cannon.
\newblock Robust adaptive model predictive control with persistent excitation conditions.
\newblock \emph{Automatica}, 152:\penalty0 110959, 2023.

\bibitem[Lu et~al.(2019)Lu, Cannon, and Koksal-Rivet]{lu2019robust}
Xiaonan Lu, Mark Cannon, and Denis Koksal-Rivet.
\newblock Robust adaptive model predictive control: Performance and parameter estimation.
\newblock \emph{International Journal of Robust and Nonlinear Control}, 2019.

\bibitem[Mania et~al.(2022)Mania, Jordan, and Recht]{mania2022active}
Horia Mania, Michael~I Jordan, and Benjamin Recht.
\newblock Active learning for nonlinear system identification with guarantees.
\newblock \emph{Journal of Machine Learning Research}, 23\penalty0 (32):\penalty0 1--30, 2022.

\bibitem[Narendra and Annaswamy(1987)]{narendra1987persistent}
Kumpati~S Narendra and Anuradha~M Annaswamy.
\newblock Persistent excitation in adaptive systems.
\newblock \emph{International Journal of Control}, 45\penalty0 (1):\penalty0 127--160, 1987.

\bibitem[Noack et~al.(2003)Noack, Afanasiev, Morzy{\'n}ski, Tadmor, and Thiele]{noack2003hierarchy}
Bernd~R Noack, Konstantin Afanasiev, Marek Morzy{\'n}ski, Gilead Tadmor, and Frank Thiele.
\newblock A hierarchy of low-dimensional models for the transient and post-transient cylinder wake.
\newblock \emph{Journal of Fluid Mechanics}, 497:\penalty0 335--363, 2003.

\bibitem[Petrov(2007)]{petrov2007lower}
Valentin~V Petrov.
\newblock On lower bounds for tail probabilities.
\newblock \emph{Journal of statistical planning and inference}, 137\penalty0 (8):\penalty0 2703--2705, 2007.

\bibitem[Sarker et~al.(2023)Sarker, Fisher, Gaudio, and Annaswamy]{sarker2023accurate}
Arnab Sarker, Peter Fisher, Joseph~E Gaudio, and Anuradha~M Annaswamy.
\newblock Accurate parameter estimation for safety-critical systems with unmodeled dynamics.
\newblock \emph{Artificial Intelligence}, page 103857, 2023.

\bibitem[Sastry and Bodson(2011)]{sastry2011adaptive}
Shankar Sastry and Marc Bodson.
\newblock \emph{Adaptive control: stability, convergence and robustness}.
\newblock Courier Corporation, 2011.

\bibitem[Sattar and Oymak(2022)]{sattar2022non}
Yahya Sattar and Samet Oymak.
\newblock Non-asymptotic and accurate learning of nonlinear dynamical systems.
\newblock \emph{Journal of Machine Learning Research}, 23\penalty0 (140):\penalty0 1--49, 2022.

\bibitem[Sattar et~al.(2022)Sattar, Oymak, and Ozay]{sattar2022finite}
Yahya Sattar, Samet Oymak, and Necmiye Ozay.
\newblock Finite sample identification of bilinear dynamical systems.
\newblock In \emph{2022 IEEE 61st Conference on Decision and Control (CDC)}, pages 6705--6711. IEEE, 2022.

\bibitem[Shi et~al.(2021)Shi, Azizzadenesheli, O'Connell, Chung, and Yue]{shi2021meta}
Guanya Shi, Kamyar Azizzadenesheli, Michael O'Connell, Soon-Jo Chung, and Yisong Yue.
\newblock Meta-adaptive nonlinear control: Theory and algorithms.
\newblock \emph{Advances in Neural Information Processing Systems}, 34:\penalty0 10013--10025, 2021.

\bibitem[Siciliano et~al.(2010)Siciliano, Sciavicco, Villani, and Oriolo]{siciliano2010robotics}
B.~Siciliano, L.~Sciavicco, L.~Villani, and G.~Oriolo.
\newblock \emph{Robotics: Modelling, Planning and Control}.
\newblock Advanced Textbooks in Control and Signal Processing. Springer London, 2010.
\newblock ISBN 9781846286414.
\newblock URL \url{https://books.google.com/books?id=jPCAFmE-logC}.

\bibitem[Simchowitz and Foster(2020)]{simchowitz2020naive}
Max Simchowitz and Dylan Foster.
\newblock Naive exploration is optimal for online lqr.
\newblock In \emph{International Conference on Machine Learning}, pages 8937--8948. PMLR, 2020.

\bibitem[Simchowitz et~al.(2018)Simchowitz, Mania, Tu, Jordan, and Recht]{simchowitz2018learning}
Max Simchowitz, Horia Mania, Stephen Tu, Michael~I Jordan, and Benjamin Recht.
\newblock Learning without mixing: Towards a sharp analysis of linear system identification.
\newblock In \emph{Conference On Learning Theory}, pages 439--473. PMLR, 2018.

\bibitem[Simpson-Porco et~al.(2016)Simpson-Porco, D{\"o}rfler, and Bullo]{simpson2016voltage}
John~W Simpson-Porco, Florian D{\"o}rfler, and Francesco Bullo.
\newblock Voltage stabilization in microgrids via quadratic droop control.
\newblock \emph{IEEE Transactions on Automatic Control}, 62\penalty0 (3):\penalty0 1239--1253, 2016.

\bibitem[Skantze et~al.(2000)Skantze, Koji{\'c}, Loh, and Annaswamy]{skantze2000adaptive}
Fredrik~P Skantze, A~Koji{\'c}, A-P Loh, and Anuradha~M Annaswamy.
\newblock Adaptive estimation of discrete-time systems with nonlinear parameterization.
\newblock \emph{Automatica}, 36\penalty0 (12):\penalty0 1879--1887, 2000.

\bibitem[Slotine and Li(1991)]{slotine1991applied}
Jean-Jacques~E Slotine and Weiping Li.
\newblock \emph{Applied nonlinear control}, volume 199.
\newblock Prentice hall Englewood Cliffs, NJ, 1991.

\bibitem[Tang et~al.(2024)Tang, Lasserre, and Yang]{tang2024uncertainty}
Yukai Tang, Jean-Bernard Lasserre, and Heng Yang.
\newblock Uncertainty quantification of set-membership estimation in control and perception: Revisiting the minimum enclosing ellipsoid.
\newblock In \emph{6th Annual Learning for Dynamics \& Control Conference}, pages 286--298. PMLR, 2024.

\bibitem[Wagenmaker and Jamieson(2020)]{wagenmaker2020active}
Andrew Wagenmaker and Kevin Jamieson.
\newblock Active learning for identification of linear dynamical systems.
\newblock In \emph{Conference on Learning Theory}, pages 3487--3582. PMLR, 2020.

\bibitem[Xu and Li(2024)]{xu2024convergence}
Haonan Xu and Yingying Li.
\newblock On the convergence rates of set membership estimation of linear systems with disturbances bounded by general convex sets.
\newblock \emph{arXiv preprint arXiv:2406.00574}, 2024.

\bibitem[Yeh et~al.(2024)Yeh, Yu, Shi, and Wierman]{yeh2024online}
Christopher Yeh, Jing Yu, Yuanyuan Shi, and Adam Wierman.
\newblock Online learning for robust voltage control under uncertain grid topology.
\newblock \emph{IEEE Transactions on Smart Grid}, 2024.

\bibitem[Yu et~al.(2023)Yu, Ho, and Wierman]{yu2023online}
Jing Yu, Dimitar Ho, and Adam Wierman.
\newblock Online adversarial stabilization of unknown networked systems.
\newblock \emph{Proceedings of the ACM on Measurement and Analysis of Computing Systems}, 7\penalty0 (1):\penalty0 1--43, 2023.

\bibitem[Ziemann and Tu(2022)]{ziemann2022learning}
Ingvar Ziemann and Stephen Tu.
\newblock Learning with little mixing.
\newblock \emph{Advances in Neural Information Processing Systems}, 35:\penalty0 4626--4637, 2022.

\bibitem[Ziemann et~al.(2023)Ziemann, Tsiamis, Lee, Jedra, Matni, and Pappas]{ziemann2023tutorial}
Ingvar Ziemann, Anastasios Tsiamis, Bruce Lee, Yassir Jedra, Nikolai Matni, and George~J Pappas.
\newblock A tutorial on the non-asymptotic theory of system identification.
\newblock In \emph{2023 62nd IEEE Conference on Decision and Control (CDC)}, pages 8921--8939. IEEE, 2023.

\bibitem[Ziemann et~al.(2024)Ziemann, Tu, Pappas, and Matni]{ziemann2024sharp}
Ingvar Ziemann, Stephen Tu, George~J Pappas, and Nikolai Matni.
\newblock Sharp rates in dependent learning theory: Avoiding sample size deflation for the square loss.
\newblock \emph{arXiv preprint arXiv:2402.05928}, 2024.

\end{thebibliography}

\newpage
\begin{appendix}

\section*{\centering Appendix}
\section*{Roadmap}
\begin{itemize}
    \item Appendix A provides a proof of Theorem~\ref{th:BMSB-o-l}. 
    \item Appendix B provides proofs of Theorem~\ref{th:lse-con-o-l} and Corollary~\ref{cor:lse-con-c-l}. 
    \item Appendix C presents a proof of Theorem~\ref{th:sme-con-o-l} and Corollaries~\ref{cor:sme-rate} and~\ref{cor:sme-con-c-l}.
    \item Appendix D provides more details of the simulation settings. 
    \item Appendix E discusses the numerical estimation of the BMSB parameters $(s_{\phi}, p_{\phi})$ in Theorem~\ref{th:BMSB-o-l}.
    \item The NeurIPS Paper Checklist is provided after the appendices. 
\end{itemize}

\section{Proof Theorem~\ref{th:BMSB-o-l}}\label{appendix:A}

\begin{proof}
Given that $u_{t} = \eta_{t}$ and satisfies the conditions in Assumption~\ref{ass: bounded iid semi continuous inputs}, $u_{t}$ is bounded, meaning $u_{t} \in \mathcal{U}$, where $\mathcal{U}$ is a compact set. Moreover, since the system is LISS, there exist functions $\gamma \in \mathcal{K}$ and $\beta \in \mathcal{KL}$, such that for all $t \geq 0$, the following holds:
$$
x_{t} \in \mathcal{X} = \bigg\{x \in \mathbb{R}^{n} : \|x\|_{2} \leq \beta(\rho_{x}, 0) + \gamma(\rho)\bigg\}
$$
with parameters $\rho_x \geq \|x_0\|_{2}$ and $\rho \geq \max(w_{\max}, u_{\max})$. Let $\mathcal{Z} = \mathcal{X} \times \mathcal{U}$, then $z_{t} \in \mathcal{Z}$ for all $t \geq 0$. The set $\mathcal{Z}$ is a compact subset of $\mathbb{R}^{n_{x}+n_{u}}$.

To show that the $\{\mathcal{F}_{t}\}_{t \geq 1}$-adapted process $\{\phi(z_{t})\}_{t \geq 1}$ satisfies the BMSB condition, it is sufficient to demonstrate that there exist $s_{\phi} > 0$ and $p_{\phi} \in (0, 1)$ such that for all $t \geq 0$ and for any $v \in \mathbb{R}^{n_{\phi}}$ with $\|v\|_{2} = 1$, the following holds: 
\begin{align}\label{eq:eq:goal-bmsb}
    {\mathbb{P}}\bigg(|v^T\phi(z_{t+1})| \geq s_{\phi} \|v\|_{2}\ \big|\ {\mathcal{F}}_{t}\bigg) \geq p_{\phi}. 
\end{align}

To establish this, we apply the Paley-Zygmund inequality, which gives a lower bound on the tail probability of a non-negative random variable: \begin{lemma}\label{lem:p-z}(Paley-Zygmund~\citep{petrov2007lower}) 
Let $x$ be a non-negative random variable. Then for any $r \in (0, 1)$, the following holds: 
\begin{align*} 
    \mathbb{P}\bigg(x > r \sqrt{\mathbb{E}[x^{2}]} \bigg) \geq (1 - r^{2})^{2} \dfrac{\mathbb{E}[x^{2}]^{2}}{\mathbb{E}[x^{4}]}. \end{align*} 
\end{lemma}

Based on this result, for any $r \in (0, 1)$, we have: \begin{align}\label{eq:pz-quad}
    \mathbb{P}\Bigg( \big|v^{\intercal}\phi(z_{t+1})\big| > r \sqrt{\mathbb{E}\bigg[\big(v^{\intercal}\phi(z_{t+1})\big)^{2}\ \big|\ \mathcal{F}_{t} \bigg]}\ \ \Bigg|\ \mathcal{F}_{t} \Bigg) \geq (1 -r^{2})^{2}\dfrac{\mathbb{E}\bigg[\big(v^{\intercal}\phi(z_{t+1})\big)^{2}\ \big|\ \mathcal{F}_{t} \bigg]^{2}}{\mathbb{E}\bigg[\big(v^{\intercal}\phi(z_{t+1})\big)^{4}\ \big|\ \mathcal{F}_{t} \bigg]}. 
\end{align}

Let $\mathcal{V} = \{v \in {\mathbb{R}}^{n_{\phi}}:\|v\|_{2}=1\}$. To show that the BMSB condition holds, it is sufficient to establish the following two points: 

\begin{itemize} 
    \item $\underset{\mathcal{F}_{t},\ t\geq 0}{\inf}\ \ \underset{\substack{v \in \mathcal{V}}}{\inf}\ \mathbb{E}\big[\big(v^{\intercal}\phi(z_{t+1})\big)^{2}\ \big|\ \mathcal{F}_{t} \big] > 0$,
    \item and $\underset{\mathcal{F}_{t},\ t\geq 0}{\sup}\ \ \underset{\substack{v \in \mathcal{V}}}{\sup}\ \mathbb{E}\big[\big(v^{\intercal}\phi(z_{t+1})\big)^{4}\ \big|\ \mathcal{F}_{t} \big] < \infty$. 
\end{itemize}

These conditions ensure that the $\{\mathcal{F}_{t}\}_{t \geq 1}$-adapted process $\{\phi(z_{t})\}_{t \geq 1}$ satisfies the BMSB condition with some constants $s_{\phi} > 0$ and $p_{\phi} \in (0, 1)$. We will divide the proof into two parts:
\\

\textbf{Step 1. Showing that $\underset{\mathcal{F}_{t},\ t\geq 0}{\inf}\ \  \underset{\substack{v \in \mathcal{V}}}{\inf}\  \mathbb{E}\big[\big(v^{\intercal}\phi(z_{t+1})\big)^{2}\ \big|\ \mathcal{F}_{t} \big] > 0$:}\\

We begin by noting the following:
\begin{align*}
    \underset{\mathcal{F}_{t},\ t\geq 0}{\inf}\ \  \underset{\substack{v \in \mathcal{V}}}{\inf}\  \mathbb{E}\bigg[\big(v^{\intercal}\phi(z_{t+1})\big)^{2}\ \big|\ \mathcal{F}_{t} \bigg] &= \underset{\mathcal{F}_{t},\ t\geq 0}{\inf}\ \  \underset{\substack{v \in \mathcal{V}}}{\inf}\  \mathbb{E}\bigg[\big(v^{\intercal}\phi\big(x_{t+1}, u_{t+1}\big)\big)^{2}\ \big|\ \mathcal{F}_{t} \bigg]\\
    &= \underset{\mathcal{F}_{t},\ t\geq 0}{\inf}\ \  \underset{\substack{v \in \mathcal{V}}}{\inf}\  \mathbb{E}\bigg[\big(v^{\intercal}\phi\big(\theta_{*}\phi(z_{t}) + w_{t}, u_{t+1}\big)\big)^{2}\ \big|\ \mathcal{F}_{t} \bigg].
\end{align*}

Since $z_{t} \in \mathcal{F}_{t}$ while $w_{t}, u_{t+1} \not \in \mathcal{F}_{t}$, we can treat $z_{t}$ as a constant and $w_{t}, u_{t+1}=\eta_{t+1}$ as random variables. From the continuity of features $\phi(\cdot)$, we can conclude that:
\begin{align*}
    \underset{\mathcal{F}_{t},\ t\geq 0}{\inf}\ \  \underset{\substack{v \in \mathcal{V}}}{\inf}\  \mathbb{E}\bigg[\big(v^{\intercal}\phi(z_{t+1})\big)^{2}\ \big|\ \mathcal{F}_{t} \bigg] &= \underset{z \in \mathcal{Z}}{\inf}\ \  \underset{\substack{v \in \mathcal{V}}}{\inf}\  \mathbb{E}\bigg[\big(v^{\intercal}\phi\big(\underbrace{\theta_{*}\phi(z) + w}_\textit{=: h(z,w)}, \eta\big)\big)^{2} \bigg],
\end{align*}
where $w$, $\eta$ are independent random variables, as assumed in Assumptions~\ref{ass: bounded iid semi continuous noises} and~\ref{ass: bounded iid semi continuous inputs}. Now, let $\mathcal{N}_{v}^{z} = \big\{(w, \eta) \in \mathcal{W} \times \mathcal{U}:\ v^{\intercal} \phi\big( h(z, w), \eta\big) = 0\big\}$, and we have:
\begin{align*}
    \mathbb{E}\bigg[\big(v^{\intercal} \phi\big( h(z, w), \eta\big)\big)^{2} \bigg] &=  \underbrace{\mathbb{E}\bigg[\big(v^{\intercal} \phi\big( h(z, w), \eta\big)\big)^{2} \mathds{1}\big\{v^{\intercal} \phi\big( h(z, w), \eta\big) = 0\big\} \bigg]}_\mathit{=0}\\
    &\ \ \ + \mathbb{E}\bigg[\big(v^{\intercal} \phi\big( h(z, w), \eta\big)\big)^{2} \mathds{1}\big\{v^{\intercal} \phi\big( h(z, w), \eta\big) \neq 0 \big\} \bigg]\\
    &= \mathbb{E}\bigg[\big(v^{\intercal} \phi\big( h(z, w), \eta\big)\big)^{2}\ \big |\ (w, \eta) \not \in \mathcal{N}_{v}^{z} \bigg] \mathbb{P}\bigg((w, \eta) \not \in \mathcal{N}_{v}^{z} \bigg)\\
    &= \mathbb{E}\bigg[\big(v^{\intercal} \phi\big( h(z, w), \eta\big)\big)^{2}\ \big |\ (w, \eta) \not \in \mathcal{N}_{v}^{z} \bigg] \bigg(1 - \mathbb{P}\bigg((w, \eta) \in \mathcal{N}_{v}^{z} \bigg)\bigg).
\end{align*}
Therefore, we have:
\begin{equation}\label{eq:2-moment}
\begin{aligned}
    \underset{z \in \mathcal{Z}}{\inf}\ \  \underset{\substack{v \in \mathcal{V}}}{\inf}\ \mathbb{E}\bigg[\big(v^{\intercal} \phi\big( h(z, w), \eta\big)\big)^{2} \bigg] &= \underset{z \in \mathcal{Z}}{\inf}\ \  \underset{\substack{v \in \mathcal{V}}}{\inf}\ \mathbb{E}\bigg[\big(v^{\intercal} \phi\big( h(z, w), \eta\big)\big)^{2}\ \big |\ (w, \eta) \not \in \mathcal{N}_{v}^{z} \bigg]\\
    &\ \ \ \ \ \times \bigg(1 - \underset{z \in \mathcal{Z}}{\sup}\ \  \underset{\substack{v \in \mathcal{V}}}{\sup}\ \mathbb{P}\bigg((w, \eta) \in \mathcal{N}_{v}^{z} \bigg) \bigg).
\end{aligned}
\end{equation}
It is evident that if $\mathcal{N}_{v}^{z} = \emptyset$, then $$\underset{z \in \mathcal{Z}}{\inf}\ \  \underset{\substack{v \in \mathcal{V}}}{\inf}\ \mathbb{E}\bigg[\big(v^{\intercal} \phi\big( h(z, w), \eta\big)\big)^{2} \bigg] \not = 0,\ \ \hbox{and}\ \ \underset{z \in \mathcal{Z}}{\sup}\ \  \underset{\substack{v \in \mathcal{V}}}{\sup}\ \mathbb{P}\bigg((w, \eta) \in \mathcal{N}_{v}^{z} \bigg) = 0,$$ leading to $\underset{z \in \mathcal{Z}}{\inf}\ \  \underset{\substack{v \in \mathcal{V}}}{\inf}\ \mathbb{E}\bigg[\big(v^{\intercal} \phi\big( h(z, w), \eta\big)\big)^{2} \bigg] > 0$. Now we proceed with the case where $\mathcal{N}_{v}^{z} \not = \emptyset$. For this, we can use the following lemma concerning the zero set of real-analytic functions in terms of Lebesgue measure.
\begin{lemma}[The zero set of real-analytic functions~\citep{cruaciun2024stability}]\label{lem:meas-zero} The set of zeros of a non-trivial real-analytic function $f:\mathbb{R}^{n} \rightarrow \mathbb{R}$ has a Lebesgue measure zero in $\mathbb{R}^{n}$.   
\end{lemma}
This is a known result and can be proved using the identity theorem along with Fubini's theorem. For further information on this topic, see sources such as~\citep{krantz2002primer,bogachev2007measure}.

Recall that we defined $h(z,w) = \theta_{*}\phi(z) + w$. Notice that $h(\cdot,\cdot)$ is real-analytic. Now consider 
\begin{align*}
    v^{\intercal} \phi\big(h(z, w), \eta\big) = \sum_{i=1}^{n_{\phi}} v^{i} \phi^{i}\big(h(z, w), \eta\big),
\end{align*}
where $\phi^{i}\big(h(z, w), \eta\big)$ are linearly independent. Hence, the sum $\sum_{i=1}^{n_{\phi}} v^{i} \phi^{i}\big(h(z, w), \eta\big) \not \equiv 0$ for any $v \in \mathcal{V}$. This implies that $v^{\intercal} \phi\big(h(z, w), \eta\big)$ is real-analytic and non-zero. Consequently, by Lemma~\ref{lem:meas-zero}, $\lambda^{n_{x}+n_{u}}(\mathcal{N}_{v}^{z}) = 0$ for any $v \in \mathcal{V}$. 

Under Assumptions~\ref{ass: bounded iid semi continuous inputs} and~\ref{ass: bounded iid semi continuous noises}, there cannot exist a set $\mathcal{E} \subset \mathcal{Z}$ of Lebesgue measure zero in $\mathbb{R}^{n_{x}+n_{u}}$ for which the $\mathbb{P}\big( (w, \eta) \in \mathcal{E}\big) = 1$. Taking this into account, along with the fact that \(\lambda^{n_{x}+n_{u}}(\mathcal{N}_{v}^{z}) = 0\) and that the sets \(\mathcal{V}\) and \(\mathcal{Z}\) are closed sets (implying they include all their limit points), we can conclude that $$\underset{z \in \mathcal{Z}}{\sup}\ \underset{\substack{v \in \mathcal{V}}}{\sup}\ \mathbb{P}\bigg((w, \eta) \in \mathcal{N}_{v}^{z} \bigg) \neq 1.$$ 

Moreover, since $\lambda^{n_{x}+n_{u}}(\mathcal{N}_{v}^{z}) = 0$ and $\lambda^{n_{x}+n_{u}}(\mathcal{W}\times\mathcal{U}) \not = 0$, it follows that $(\mathcal{N}_{v}^{z})^{c} \not = \emptyset$. This implies 
$$\underset{z \in \mathcal{Z}}{\inf}\ \  \underset{\substack{v \in \mathcal{V}}}{\inf}\ \mathbb{E}\bigg[\big(v^{\intercal} \phi\big( h(z, w), \eta\big)\big)^{2}\ \big |\ (w, \eta) \not \in \mathcal{N}_{v}^{z} \bigg] \not = 0.$$ 

Substituting these results into~\eqref{eq:2-moment}, we obtain:
\begin{align*}
    \underset{\mathcal{F}_{t},\ t\geq 0}{\inf}\ \  \underset{\substack{v \in \mathcal{V}}}{\inf}\  \mathbb{E}\big[\big(v^{\intercal}\phi(z_{t+1})\big)^{2}\ \big|\ \mathcal{F}_{t} \big] > 0.
\end{align*}
\\

\textbf{Step 2. Showing that $\underset{\mathcal{F}_{t},\ t\geq 0}{\sup}\ \  \underset{\substack{v \in \mathcal{V}}}{\sup}\  \mathbb{E}\big[\big(v^{\intercal}\phi(z_{t+1})\big)^{4}\ \big|\ \mathcal{F}_{t} \big] < \infty$:}

Since $z_{t} \in \mathcal{Z}$ for $t \geq 0$, and considering that the noise and disturbances are bounded while the features are real-analytic, it follows that $z_{t+1}|\mathcal{F}_{t}$ is a bounded random variable. Consequently, $v^{\intercal}\phi(z_{t+1})|\mathcal{F}_{t}$ is also bounded. Given that both $\mathcal{Z}$ and $\mathcal{V}$ are compact sets—meaning they contain all their limit points—and that any random variable with bounded support has finite moments, we conclude that 
$$
    \underset{\mathcal{F}_{t},\ t\geq 0}{\sup}\ \  \underset{\substack{v \in \mathcal{V}}}{\sup}\  \mathbb{E}\big[\big(v^{\intercal}\phi(z_{t+1})\big)^{4}\ \big|\ \mathcal{F}_{t} \big] < \infty.
$$

We finalize the proof by combining the results from Step 1 and Step 2.
\end{proof}

\section{Proofs for Theorem~\ref{th:lse-con-o-l} and Corollary~\ref{cor:lse-con-c-l}}\label{appendix:B}

\setcounter{section}{2}
\subsection{Proof of Theorem~\ref{th:lse-con-o-l}}\label{app:B1}

\begin{proof}
The proof hinges on the following key meta-theorem about the LSE convergence rate:

\begin{theorem}[LSE meta-theorem~\citep{simchowitz2018learning}]~\label{th:ols-meta}
Fix $\delta \in (0, 1)$, $T \geq 1$, and $0 \prec \Gamma_{sb} \prec \bar{\Gamma}$. Consider a random process $\{(y_{t}, x_{t})\}_{t\geq1} \in (\mathbb{R}^{n_{y}} \times \mathbb{R}^{n_{x}})^{T}$, and a filtration $\{\mathcal{F}_{t}\}_{t \geq 1}$. Suppose the following conditions hold:
\begin{itemize}
    \item $x_{t} = \theta_{*} y_{t} + w_{t}$, where $w_{t} | \mathcal{F}_{t}$ is a zero mean $\sigma_{w}^{2}$-sub-Gaussian,
    \item $\{y_{t}\}_{t \geq 1}$ is an $\{\mathcal{F}_{t}\}_{t \geq 1}$-adapted random process satisfying the $(k, \Gamma_{sb}, p)$-BMSB condition,
    \item $\mathbb{P}\big(\sum_{t=1}^{T} y_{t}y_{t}^{\intercal} \not \preceq T\bar{\Gamma} \big) \leq \delta$.
\end{itemize}
If the trajectory length $T$ satisfies 
\begin{align*}
    T \geq \dfrac{10k}{p^{2}} \Bigg( \log\bigg(\frac{1}{\delta}\bigg) +  \log \det \bigg( \bar{\Gamma} \Gamma_{sb}^{-1}\bigg) + 2n_{y}\log\bigg(\frac{10}{p}\bigg)\Bigg),
\end{align*}
then with probability at least $1-3\delta$, LSE estimation error is bounded by:
\begin{align*}
 \big \|\hat{\theta}_{T} - \theta_{*} \big\|_{2} \leq \dfrac{90 \sigma_{w}}{p} \sqrt{ \dfrac{n_{x} +\log\bigg(\frac{1}{\delta}\bigg) +  \log \det \bigg( \bar{\Gamma} \Gamma_{sb}^{-1}\bigg) + n_{y}\log\bigg(\frac{10}{p}\bigg)}{T\sigma_{min}(\Gamma_{sb})}}.
\end{align*}
\end{theorem}

Since $w_{t}$ satisfies the Assumption~\ref{ass: bounded iid semi continuous noises}), and $w_{t} \notin \mathcal{F}_{t}$, then $\sigma{w} | \mathcal{F}_{t}$ is sub-Gaussian with parameter $\sigma{w}$. Additionally, system~(\ref{eq:sys}) is linear in unknown parameters $\theta_{*}$. From Theorem~\ref{th:BMSB-o-l}, the $\{\mathcal{F}_{t}\}_{t \geq 1}$-adapted process $\{\phi(z_{t})\}_{t \geq 1}$ satisfies the $(1, s^{2}_{\phi}I_{n_{\phi}}, p_{\phi})$-BMSB condition for some $s_{\phi} > 0$ and $p_{\phi} \in (0, 1]$. 

To complete the proof, it is left to show that for any $\delta \in (0,1)$, there exists a $\bar{\Gamma} \succ s^{2}_{\phi}I_{n_{\phi}}$, such that  
$$\mathbb{P}\bigg( \sum_{t=1}^{T} \phi(z_{t})\phi^{\intercal}(z_{t}) \npreceq T\bar{\Gamma}\bigg) \leq \delta.$$

To see this, note that for $\bar{b}_{\phi} = \sup_{t \geq 0} \mathbb{E}\big[\|\phi(z_{t})\|^{2}_{2}\big]$, we have:
\begin{equation*}
\begin{aligned}
\mathbb{P}\bigg( \sum_{t=1}^{T} \phi(z_{t})\phi^{\intercal}(z_{t}) \npreceq \dfrac{\bar{b}_{\phi}T}{\delta}I_{n_{\phi}}\bigg) &= \mathbb{P}\Bigg(\lambda_{\max}\bigg(\sum_{t=1}^{T} \phi(z_{t})\phi^{\intercal}(z_{t}) \bigg)\succ \lambda_{\max}\bigg(\dfrac{\bar{b}_{\phi}T}{\delta}I_{n_{\phi}}\bigg)\Bigg)\\
&= \mathbb{P}\bigg(\big\|\sum_{t=1}^{T} \phi(z_{t})\phi^{\intercal}(z_{t}) \big\|_{2} \succ \dfrac{\bar{b}_{\phi}T}{\delta} \bigg)\\
&\leq \dfrac{\delta \mathbb{E}\bigg[\big\|\sum_{t=1}^{T} \phi(z_{t})\phi^{\intercal}(z_{t}) \big\|_{2} \bigg]}{\bar{b}_{\phi}T}.
\end{aligned} 
\end{equation*}

In addition, we have:
\begin{align*}
     \mathbb{E}\bigg[\big\|\sum_{t=1}^{T} \phi(z_{t})\phi^{\intercal}(z_{t})\big\|_{2} \bigg] \leq
    \sum_{t=1}^{T} \mathbb{E}\big[ \big\|\phi(z_{t})\phi^{\intercal}(z_{t})\big\|_{2} \big] \leq \sum_{t=1}^{T} \mathbb{E}\big[\|\phi(z_{t})\|^{2}_{2}\big] \leq T \sup_{t \geq 0} \mathbb{E}\big[\|\phi(z_{t})\|^{2}_{2}\big],
\end{align*}
which implies that:
\begin{equation}\label{eq:state-corr}
\begin{aligned}
\mathbb{P}\bigg( \sum_{t=1}^{T} \phi(z_{t})\phi^{\intercal}(z_{t}) \npreceq  T\bar{\Gamma}\bigg)
&\leq \delta,
\end{aligned}
\end{equation}
where $\bar{\Gamma} = \dfrac{\bar{b}_{\phi}}{\delta}I_{n_{\phi}}$. since  $z_{t} \in \mathcal{Z}$ for all $t \geq 0$, with $\mathcal{Z}$ being a compact set (due to the system's LISS property and features $\phi(\cdot)$ being real-analytic), such a bounded $\bar{b}_{\phi}$ exists, completing the proof.
\end{proof}

%%%%%%%%%%%%%%%%%%%%%%%%%%%%%%%%%%%%%%%%%%%%%%%%%%%%%%%%
\setcounter{section}{2}
\subsection{Proof of Corollary~\ref{cor:lse-con-c-l}}\label{app:B2}

\begin{proof}
We start the proof by stating the following lemma which is extension of Theorem~\ref{th:BMSB-o-l} to the case with $u_{t} = \pi(x_{t}) + \eta_{t}$. 

\begin{lemma}\label{lem:BMSB_c_l}
Let $u_t=\pi(x_t)+\eta_t$, $\pi(\cdot)$ is real-analytic, $\eta_{t}$ satisfies Assumption~\ref{ass: bounded iid semi continuous inputs}. Consider the filtration $\mathcal{F}_{t} = \mathcal{F}(w_{0}, \cdots,w_{t-1}, x_{0}, \cdots, x_{t}, \pi(x_{0}), \cdots, \pi(x_{t}), \eta_{0}, \cdots, \eta_{t})$. Suppose Assumptions~\ref{ass: analytic},~\ref{ass: bounded iid semi continuous noises} hold and that the closed-loop system $x_{t+1}=\theta_*\phi(x_t, \pi(x_t)+\eta_t)+w_t$ satisfies Assumption \ref{ass: iss system} for both $w_t$ and $\eta_t$. Then there exist $s_{\phi} > 0$, and $p_{\phi} \in (0, 1)$ such that the $\{\mathcal{F}_{t}\}_{t \geq 1}$-adapted process $\big\{\phi\big(x_{t}, u_{t}\big)\big\}_{t \geq 1}$ satisfies the
$\big(1, s_{\phi}^{2}I_{n_{\phi}}, p_{\phi}\big)$-BMSB condition. 
\end{lemma}

\begin{proof}[Proof of Lemma \ref{lem:BMSB_c_l}]
The proof closely follows the steps of Theorem~\ref{th:BMSB-o-l}. First, observe that $u_{t} = \pi(x_{t}) + \eta_{t}$. Since $\eta_{t}$ satisfies the conditions outlined in Assumption~\ref{ass: bounded iid semi continuous inputs}, and the closed-loop system $x_{t+1}=\theta_*\phi(x_t, \pi(x_t)+\eta_t)+w_t$ adheres to Assumption~\ref{ass: iss system} with respect to both $w_t$ and $\eta_t$, there exist functions $\gamma \in \mathcal{K}$ and $\beta \in \mathcal{KL}$ such that for all $t \geq 0$ the following holds:
$$
x_{t} \in \mathcal{X} = \bigg\{x \in \mathbb{R}^{n} : \|x\|_{2} \leq \beta(\rho_{x}, 0) + \gamma(\rho)\bigg\}
$$
with parameters $\rho_x \geq \|x_0\|_{2}$ and $\rho \geq \max(w_{\max}, u_{\max})$. Let $\mathcal{Z} = \mathcal{X} \times \mathcal{U}$, where $\mathcal{U}$ is a compact set containing $u_{t}$ for all $t\geq 0$. Thus, $z_{t} \in \mathcal{Z}$ for all $t \geq 0$. The set $\mathcal{Z}$ is a compact subset of $\mathbb{R}^{n_{x}+n_{u}}$. The remaining part of the proof, specifically to show that $\underset{\mathcal{F}_{t},\ t\geq 0}{\sup}\ \  \underset{\substack{v \in \mathcal{V}}}{\sup}\  \mathbb{E}\big[\big(v^{\intercal}\phi(z_{t+1})\big)^{4}\ \big|\ \mathcal{F}_{t} \big] < \infty$ follows similarly to the proof in Theorem~\ref{th:BMSB-o-l}. 

It remains to show that $\underset{\mathcal{F}_{t},\ t\geq 0}{\inf}\ \  \underset{\substack{v \in \mathcal{V}}}{\inf}\  \mathbb{E}\big[\big(v^{\intercal}\phi(z_{t+1})\big)^{2}\ \big|\ \mathcal{F}_{t} \big] > 0$. This can be shown as follows:
\begin{align*}
    \underset{\mathcal{F}_{t},\ t\geq 0}{\inf}\ \  \underset{\substack{v \in \mathcal{V}}}{\inf}\  \mathbb{E}\bigg[\big(v^{\intercal}\phi(z_{t+1})\big)^{2}\ \big|\ \mathcal{F}_{t} \bigg] &= \underset{\mathcal{F}_{t},\ t\geq 0}{\inf}\ \  \underset{\substack{v \in \mathcal{V}}}{\inf}\  \mathbb{E}\bigg[\big(v^{\intercal}\phi\big(x_{t+1}, u_{t+1}\big)\big)^{2}\ \big|\ \mathcal{F}_{t} \bigg]\\
    &= \underset{\mathcal{F}_{t},\ t\geq 0}{\inf}\ \  \underset{\substack{v \in \mathcal{V}}}{\inf}\  \mathbb{E}\bigg[\bigg(v^{\intercal}\phi\big(x_{t+1}, \pi(x_{t+1}) + \eta_{t+1}\big)\bigg)^{2}\ \big|\ \mathcal{F}_{t} \bigg]\\
    &= \underset{\mathcal{F}_{t},\ t\geq 0}{\inf}\ \  \underset{\substack{v \in \mathcal{V}}}{\inf}\  \mathbb{E}\bigg[\bigg(v^{\intercal}\phi\bigg(\theta_{*}\phi(z_{t}) + w_{t},\\
    &\ \ \ \ \ \ \ \ \ \ \ \ \ \ \ \ \ \ \ \ \ \ \ \ \ \ \ \ \ \ \ \ \ \ \ \ \ \ \ \ \ \ \ \ \ \pi\big(\theta_{*}\phi(z_{t}) + w_{t}\big) + \eta_{t+1}\bigg)\bigg)^{2}\ \big|\ \mathcal{F}_{t} \bigg].
\end{align*}
Since $z_{t} \in \mathcal{F}_{t}$ and $w_{t}, \eta_{t+1} \not \in \mathcal{F}_{t}$ then $z_{t}$ is treated as constant while $w_{t}, \eta_{t+1}$ are considered random variables. Then, by the continuity of $\phi$ we conclude that
\begin{align*}
    \underset{\mathcal{F}_{t},\ t\geq 0}{\inf}\ \  \underset{\substack{v \in \mathcal{V}}}{\inf}\  \mathbb{E}\bigg[\big(v^{\intercal}\phi(z_{t+1})\big)^{2}\ \big|\ \mathcal{F}_{t} \bigg] &= \underset{z \in \mathcal{Z}}{\inf}\ \  \underset{\substack{v \in \mathcal{V}}}{\inf}\  \mathbb{E}\bigg[\bigg(v^{\intercal}\phi\bigg(\underbrace{\theta_{*}\phi(z) + w}_\textit{=: h(z,w)}, \pi\big(\underbrace{\theta_{*}\phi(z) + w}_\textit{=: h(z,w)}\big) + \eta\bigg)\bigg)^{2} \bigg]\\
    &= \underset{z \in \mathcal{Z}}{\inf}\ \  \underset{\substack{v \in \mathcal{V}}}{\inf}\  \mathbb{E}\bigg[\bigg(v^{\intercal}\phi\bigg(h(z,w), \pi\big(h(z,w)\big) + \eta\bigg)\bigg)^{2} \bigg]
\end{align*}
where $w$ and $\eta$ are independent random variables constrained as described in Assumptions~\ref{ass: bounded iid semi continuous noises} and~\ref{ass: bounded iid semi continuous inputs}. By letting $\mathcal{N}_{v}^{z} = \big\{(w, \eta) \in \mathcal{W} \times \mathcal{U}:\ v^{\intercal} \phi\big( h(z, w), \pi\big(h(z, w)\big) + \eta\big) = 0\big\}$, similar to~(\ref{eq:2-moment}) we have:
\begin{align*}
   \underset{z \in \mathcal{Z}}{\inf}\ \  \underset{\substack{v \in \mathcal{V}}}{\inf}\  \mathbb{E}\bigg[\bigg(v^{\intercal}\phi\bigg(h(z,w), &\pi\big(h(z,w)\big) + \eta\bigg)\bigg)^{2} \bigg]\\
   &= \underset{z \in \mathcal{Z}}{\inf}\ \  \underset{\substack{v \in \mathcal{V}}}{\inf}\ \mathbb{E}\bigg[\bigg(v^{\intercal}\phi\bigg(h(z,w), \pi\big(h(z,w)\big) + \eta\bigg)\bigg)^{2}\ \big |\ (w, u) \not \in \mathcal{N}_{v}^{z} \bigg]\\
    &\ \ \ \ \ \  \times \bigg(1 - \underset{z \in \mathcal{Z}}{\sup}\ \  \underset{\substack{v \in \mathcal{V}}}{\sup}\ \mathbb{P}\bigg((w, \eta) \in \mathcal{N}_{v}^{z} \bigg) \bigg).
\end{align*}
We aim to show that $\lambda^{n_x + n_u}(\mathcal{N}_{v}^z) = 0$ by applying Lemma~\ref{lem:meas-zero}. Note that $\phi(\cdot)$, $h(\cdot)$, and $\pi(\cdot)$ are real-analytic. To use the results of Lemma~\ref{lem:meas-zero}, we need to establish that $v^{\intercal} \phi\bigg( h(z, w), \pi\big(h(z, w)\big) + \eta\bigg)$ is non-zero for any $v \in \mathcal{V}$. First, observe that: 
\begin{align*}
    v^{\intercal} \phi\big(h(z, w), \pi(h(z, w)) + \eta\big) = \sum_{i=1}^{n_{\phi}} v_{i} \phi^{i}\big(h(z, w), \pi(h(z, w)) + \eta)\big).
\end{align*}
Now consider two scenarios:
\begin{itemize}
    \item All components of  $\pi\big(h(z, w)\big)$ are linearly independent with any component of $h(z, w)$.
    \item At least one component of $\pi\big(h(z, w)\big)$ is linearly dependent with one or more components of $h(z, w)$.
\end{itemize}
In both cases, due to the additive nature of $\eta$, all the functions $\phi^{i}\big(h(z, w), \pi(h(z, w)) + \eta)\big)$ with $i=1,\cdots,n_{\phi}$ are linearly independent, ensuring that $v^{\intercal} \phi\big(h(z, w), \pi(h(z, w)) + \eta\big) \not \equiv 0$ for any $v \in \mathcal{V}$. The remainder of the proof follows similarly to the argument in Theorem~\ref{th:BMSB-o-l}.
\end{proof}

Using this Lemma, Theorem~\ref{th:ols-meta}, and reasoning similar to that in the proof of Theorem~\ref{th:lse-con-o-l}, the proof can be completed.
\end{proof}
\section{Proofs for Theorem~\ref{th:sme-con-o-l}, Corollary~\ref{cor:sme-rate}, and Corollary~\ref{cor:sme-con-c-l}}\label{appendix:C}

\setcounter{section}{3}
\subsection{Proof of Theorem~\ref{th:sme-con-o-l}}\label{app:C1}

\begin{proof}
The proof follows from applying the following meta-theorem on the convergence rate of SME.

\begin{theorem}[SME meta-theorem~\citep{li2024icml}]\label{lem:sm-meta}
Consider a general time series model with linear responses as follows:
$$x_{t} = \theta_{*} y_{t} + w_{t},\ \ \ t \geq 0.$$
Also, define the filtration $\mathcal{F}_{t} = \mathcal{F}(w_{0}, \cdots, w_{t-1}, y_{0}, \cdots, y_{t})$. Assume the following conditions are met:
\begin{itemize}
    \item $w_{t}$ are i.i.d. with variance $\sigma^{2}_{w}I_{n_{x}}$, and box-constrained, i.e., $w_{t} \in \mathcal{W} = \{w\in \mathbb{R}^{n_{x}}: \|w\|_{\infty} \leq w_{\max}\}$ for some $w_{\max} > 0$.
    \item $\{y_{t}\}_{t \geq 1}$ is an $\{\mathcal{F}_{t}\}_{t \geq 1}$-adapted random process satisfying the $(k, s_{y}^{2}I_{n_{y}}, p_{y})$-BMSB condition.
    \item There exists $b_{y} > 0$ such that $\|y_{t}\|_{2} \leq b_{y}$ almost surely for all $t \geq 0$.
    \item For any $\ell > 0$, there exists $q_{w}(\ell) > 0$, such that for any $1 \leq j \leq n$ and $t \geq 0$, we have $\mathbb{P}(w_{t}^{j} + w_{\max} \leq \ell) \geq q_{w}(\ell) > 0$, $\mathbb{P}(w_{\max} - w_{t}^{j} \leq \ell) \geq q_{w}(\ell) > 0$.
\end{itemize}
Then for any $m \geq 1$ and any $\delta \in (0, 1)$, when $T > m$, the diameter of the uncertainty set 
$$\Theta_{T} = \bigcap \limits_{t=0}^{T-1} \bigg\{\hat{\theta}: x_{t} - \hat{\theta} y_{t} \in \mathcal{W} \bigg\}$$ satisfies:
\begin{align*}
    \mathbb{P}\bigg(\textup{diam}(\Theta_{T}) > \delta\bigg) 
    &\leq 544 \frac{T}{m} n_{y}^{2.5} \log(a_{2}n_{y})a_{2}^{n_{y}}\exp(-a_{3}m)\\
    &\ \ \ + 544 n_{x}^{2.5} n_{y}^{2.5} \log(a_{4} n_{x} n_{y}) a_{4}^{n_{x}n_{y}}\bigg(1-q_{w} \bigg(\dfrac{a_{1}\delta}{4\sqrt{n_{x}}}\bigg)\bigg)^{\lceil T/ m \rceil},
\end{align*}
where $a_{1} = \frac{s_{y}p_{y}}{4}$, $a_{2} = \frac{64b_{y}^{2}}{s_{y}^{2}p_{y}^{2}}$, $a_{3} = \frac{p_{y}^{2}}{8}$, $a_{4} = \max\bigg(\frac{4b_{y}\sqrt{n_{x}}}{a_{1}}, 1\bigg)$.
\end{theorem}

Observe that system~(\ref{eq:sys}) is linear in the unknown parameters $\theta_{*}$, and we can prove Theorem~\ref{th:sme-con-o-l} by showing that the $\{\mathcal{F}_{t}\}_{t \geq 1}$-adapted process $\{\phi(z_{t})\}_{t \geq 1}$ and $w_{t}$ meet the conditions of the meta-theorem. By Assumptions~\ref{ass: bounded iid semi continuous noises} and~\ref{ass:w-tight}, and since $w_{t} \not \in \mathcal{F}_{t}$, the noise $w_{t}$ fulfills all the requirements of the meta-theorem. Moreover, according to Theorem~\ref{th:BMSB-o-l}, the $\{\mathcal{F}_{t}\}_{t \geq 1}$-adapted process ${\phi(z_{t})}_{t \geq 1}$ satisfies the $(1, s^{2}_{\phi}I_{n_{\phi}}, p_{\phi})$-BMSB condition for some $s_{\phi} > 0$ and $p_{\phi} \in (0, 1]$. Lastly, since the system is LISS, we have $z_{t} \in \mathcal{Z}$ for all $t \geq 0$, where $\mathcal{Z}$ is the compact set defined in the proof of Theorem~\ref{th:BMSB-o-l}. Therefore, there exists a constant $b_{\phi} > 0$ such that $\sup_{t\geq 0} \|\phi(z_{t})\|_{2} \leq b_{\phi}$, completing the proof of the theorem.

Explicitly, this means that for any $m \geq 1$, for any $\delta \in (0, 1)$, when $T > m$, we have:
\begin{equation}\label{eq:prob-diam}
\begin{aligned}
    \mathbb{P}\bigg(\textup{diam}(\Theta_{T}) > \delta\bigg) 
    &\leq 544 \frac{T}{m} n_{\phi}^{2.5} \log(a_{2}n_{\phi})a_{2}^{n_{\phi}}\exp(-a_{3}m)\\
    &\ \ \ + 544 n_{x}^{2.5} n_{\phi}^{2.5} \log(a_{4} n_{x} n_{\phi}) a_{4}^{n_{x}n_{\phi}}\bigg(1-q_{w} \bigg(\frac{a_{1}\delta}{4\sqrt{n_{x}}}\bigg)\bigg)^{\lceil T/ m \rceil},
\end{aligned}
\end{equation}
where $a_{1} = \frac{s_{\phi}p_{\phi}}{4}$, $a_{2} = \frac{64b_{\phi}^{2}}{s_{\phi}^{2}p_{\phi}^{2}}$, $a_{3} = \frac{p_{\phi}^{2}}{8}$, $a_{4} = \frac{16b_{\phi}\sqrt{n_{x}}}{s_{\phi}p_{\phi}}$. 
\end{proof}

\subsection{Proof of Corollary~\ref{cor:sme-rate}}\label{app:C2}
Let us provide two example distributions, truncated Gaussian and uniform, along with their corresponding $q_{w}(\cdot)$ (from~\citep{li2024icml}):

\begin{itemize} 
    \item If $w_{t}$ follows a uniform distribution on $[-w_{\max}, w_{\max}]^{n_{x}}$, then $q_{w}(\ell) = c_{w} \ell$ with $c_{w} = \frac{1}{2w_{\max}}$.     
    \item If $w_{t}$ follows a truncated-Gaussian distribution on $[-w_{\max}, w_{\max}]^{n_{x}}$, generated by a Gaussian distribution with zero mean and covariance matrix $\sigma_{w}^{2}I_{n_{x}}$, then $q_{w}(\ell) = c_{w} \ell$ with $c_{w} = \frac{1}{\min(\sqrt{2\pi}\sigma_{w}, 2w_{\max})} \exp(\frac{-w_{\max}^{2}}{2\sigma_{w}^{2}})$.
\end{itemize}

Now, fix $\epsilon \in (0, 1)$. We want to show that if $q\big(\frac{a_{1}\delta}{4\sqrt{n_{\phi}}}\big) = c_{w}\frac{a_{1}\delta}{4\sqrt{n_{\phi}}}$ and we choose $m \geq 1$ such that
\begin{align}\label{eq:m}
    m \geq \dfrac{1}{a_{3}} \bigg( \log\bigg(\frac{T}{\epsilon}\bigg) + n_{\phi} \log(a_{2}) + 2.5 \log(n_{\phi}) + \log \log (a_{2} n_{\phi}) + \log(544) \bigg),
\end{align}
then for all $T \geq m$, we have
\begin{align}\label{eq:delta}
    \delta \leq O\Bigg( \dfrac{\sqrt{n_{x}} \log\big(\frac{1}{\epsilon}\big) + n_{x}^{1.5}n_{\phi}\log\big(\frac{b_{\phi}\sqrt{n_{x}}}{s_{\phi}p_{\phi}}\big)}{c_{w}s_{\phi}p_{\phi}T} \Bigg)
\end{align}
with probability at least $1-2\epsilon$. 

Let the two terms in right hand-side of~\eqref{eq:prob-diam} be denoted by "term 1" and "term 2". We proceed with the proof in two steps as follows:

\textbf{Step 1: showing that with choice of $m$ in~\eqref{eq:m}, term 1 $\leq \epsilon$ :} 

With this choice of $m$, it is straightforward to see that
\begin{align*}
    544 T n_{\phi}^{2.5} \log(a_{2}n_{\phi})a_{2}^{n_{\phi}}\exp(-a_{3}m) \leq \epsilon,
\end{align*}
and thus term 1 $\leq \epsilon$.

\textbf{Step 2: letting term 2 $= \epsilon$ and showing that $\delta$ satisfies the inequality in~\eqref{eq:delta}:} 

Assuming without loss of generality that $\frac{T}{m}$ is an integer, note that term 2 $= \epsilon$ implies:
\begin{align*}
    q_{w} \bigg(\dfrac{a_{1}\delta}{4\sqrt{n_{x}}}\bigg) &=  \bigg( 1- \bigg(\dfrac{\epsilon}{544 n_{x}^{2.5} n_{\phi}^{2.5} \log(a_{4} n_{x} n_{\phi}) a_{4}^{n_{x}n_{\phi}}} \bigg)^{\frac{m}{T}} \bigg)\\
    &\leq - \log \bigg(\dfrac{\epsilon}{544 n_{x}^{2.5} n_{\phi}^{2.5} \log(a_{4} n_{x} n_{\phi}) a_{4}^{n_{x}n_{\phi}}} \bigg)^{\frac{m}{T}}\\
    &= - \frac{m}{T} \log \bigg(\dfrac{\epsilon}{544 n_{x}^{2.5} n_{\phi}^{2.5} \log(a_{4} n_{x} n_{\phi}) a_{4}^{n_{x}n_{\phi}}} \bigg)\\
    &= \frac{m}{T} \bigg( \log\bigg(\frac{1}{\epsilon}\bigg) + \log(a_{4}) n_{x}n_{\phi}  + 2.5\log(n_{x} n_{\phi})  + \log \log (a_{4} n_{x} n_{\phi}) + \log(544) \bigg).
\end{align*}

If $q_{w} \big(\frac{a_{1}\delta}{4\sqrt{n_{x}}}\big) = c_{w} \frac{a_{1}\delta}{4\sqrt{n_{x}}}$ for some constant $c_{w} > 0$, then:
\begin{align*}
    \delta &\leq  
    \frac{4 \sqrt{n_{x}}m}{c_{w} a_{1} T} \bigg( \log\bigg(\frac{1}{\epsilon}\bigg) + \log(a_{4}) n_{x}n_{\phi}  + 2.5\log(n_{x} n_{\phi})  + \log \log (a_{4} n_{x} n_{\phi}) + \log(544) \bigg)\\
    &\leq \frac{16 \sqrt{n_{x}}m}{c_{w} s_{\phi} p_{\phi} T} O\Bigg(  \log\bigg(\frac{1}{\epsilon}\bigg) +  n_{x}n_{\phi}\log\bigg(\frac{16 b_{\phi}\sqrt{n_{x}}}{s_{\phi}p_{\phi}}\bigg) \Bigg)\\
    &\leq O\Bigg( \dfrac{\sqrt{n_{x}} \log\big(\frac{1}{\epsilon}\big) + n_{x}^{1.5}n_{\phi}\log\big(\frac{b_{\phi}n_{x}}{s_{\phi}p_{\phi}}\big)}{c_{w}s_{\phi}p_{\phi}T} \Bigg).
\end{align*}

Combining these two steps, we conclude that, with probability at least $1-2\epsilon$,
\begin{align*}
    \textup{diam}(\Theta_{T}) \leq 
    O\Bigg( \dfrac{\sqrt{n_{x}} \log\big(\frac{1}{\epsilon}\big) + n_{x}^{1.5}n_{\phi}\log\big(\frac{b_{\phi}n_{x}}{s_{\phi}p_{\phi}}\big)}{c_{w}s_{\phi}p_{\phi}T} \Bigg).
\end{align*}

\subsection{Proof of Corollary~\ref{cor:sme-con-c-l}}\label{app:C3}
This corollary's proof builds on Lemma~\ref{lem:BMSB_c_l} in Appendix~\ref{app:B2} and closely aligns with the proofs of Theorem~\ref{th:sme-con-o-l} and Corollary~\ref{cor:sme-rate}.

\section{Numerical Experiments}\label{appendix:D}

This section provides details on the simulation experiments. 

\setcounter{section}{4}
\subsection{Pendulum}\label{app:pend}
The ground truth for the unknown parameters for pendulum example in Example~\ref{example: pendulum} is set to be
\begin{align*}
    m = 0.1 \ \hbox{(kg)},\ \ l = 0.5 \ \hbox{(m)},
\end{align*}
and discretization time step in our numerical experiments is  $dt = 0.01 \ \hbox{(s)}$. The control input is a simple feedback controller $u_{t} = - k \dot{\alpha}_{t} + \eta_{t}$. In Figures~\ref{fig:lse_pend_uni} and ~\ref{fig:lse_pend_trunc} we choose $k=2$ and in Figures~\ref{fig:sme_pend_uni}, ~\ref{fig:sme_pend_trunc} and~\ref{fig:pend_sme} we choose $k=0.1$. Note there are two unknown parameters in this pendulum example as follows:
\begin{align*}
    \theta_{1} = \frac{1}{l},\ \ \theta_{2} = \frac{1}{ml^{2}}.
\end{align*}

\subsection{Quadrotor}\label{app:quad}
The ground truth for the unknown parameters for quadrotor example in Example~\ref{example: quadrotor} is set to be
\begin{align*}
    m &= 0.468\ (\text{kg}),\\
    I_{xx} = 4.856\times 10^{-3}\ (\text{kg/m}^2),\ \ 
    I_{yy} &= 4.856\times 10^{-3}\ (\text{kg/m}^2),\ \  
    I_{zz} = 8.801\times 10^{-3}\ (\text{kg/m}^2).     
\end{align*}
The discretization time step in our numerical experiments is  $dt = 0.01 \ \hbox{(s)}$. The control input is a control policy plus i.i.d. noise. The control policy on altitude and the three Euler angles is borrowed from~\citep{alaimo2013mathematical}. The controller gains in our numerical experiments are chosen as: 
\begin{align*}
    &kp_z = 0.75,\ \ kd_z = 1.25,\\
    &kp_{\phi} = 0.03,\ \ kd_{\phi} = 0.00875,\\ 
    &kp_{\theta} = 0.03,\ \ kd_{\theta} = 0.00875,\\
    &kp_{\psi} = 0.03,\ \ kd_{\psi} = 0.00875.
\end{align*}

Note that there are seven unknown parameters in this quadrotor example, as follows:
\begin{align*}
    \theta_{1} &= \frac{1}{m},\\
    \theta_{2} = \frac{1}{I_{xx}},\ \ \ \theta_{3} = \frac{I_{yy}-I_{zz}}{I_{xx}},\ \ \theta_{4} = \frac{1}{I_{yy}}&,\ \ \ \theta_{5} = \frac{I_{zz}-I_{xx}}{I_{yy}},\ \ \theta_{6} = \frac{1}{I_{zz}},\ \ \ \theta_{7} = \frac{I_{xx}-I_{zz}}{I_{zz}}.
\end{align*}

Figure~\ref{fig:sm_quadrotor_2} displays the uncertainty set estimated by SME for the seven unknown parameters in the quadrotor example for various trajectory lengths, with \(\eta_{t}\) and \(w_{t}\) being i.i.d. samples from truncated-Gaussian distributions. The uncertainty sets are observed to shrink as the trajectory length increases, consistent with our theoretical results. Note that the ground truth value is contained within all the uncertainty sets.
\begin{figure}[ht]
    \centering    \includegraphics[width=1\linewidth]{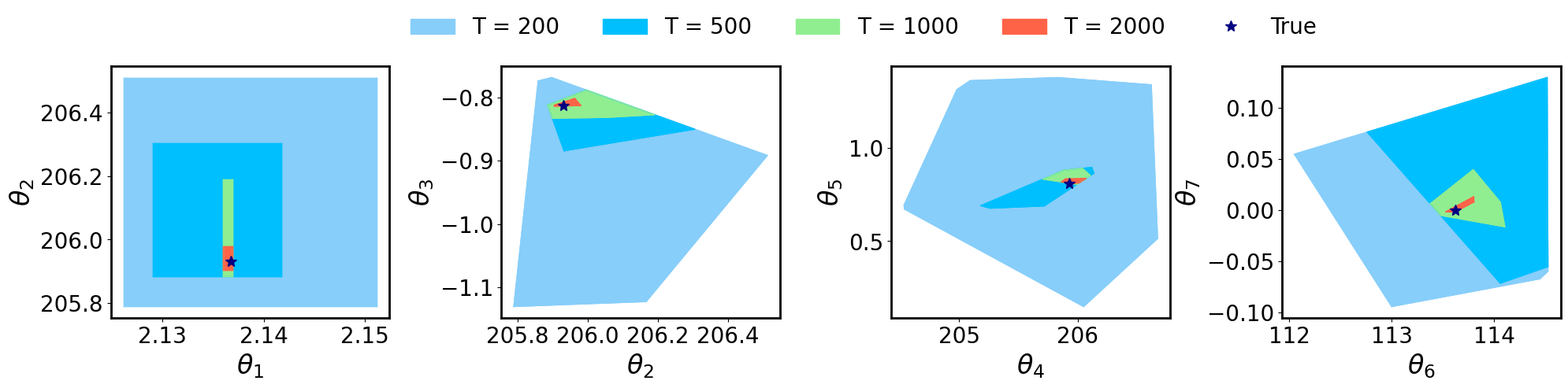}
    \caption{\footnotesize 2D projections of the uncertainty set estimated by SME for the unknown parameters of the quadrotor example. The noises and disturbances are i.i.d generated from $\texttt{truncated-Gaussian}(0,0.5,[-1,1])$. 
    }
    \label{fig:sm_quadrotor_2}
\end{figure}

\section{Numerical Estimation of BMSB Parameters $(s_{\phi}, p_{\phi})$}\label{appendix:E}

We compare the empirical rates of LSE and SME with their theoretical counterparts in Section~\ref{sec:num-exp}. The theoretical results presented in Theorem~\ref{th:lse-con-o-l} and Corollary~\ref{cor:sme-rate} rely on the parameters \(b_{\phi}\), \(\bar{b}_{\phi}\), and the BMSB parameters \((s_{\phi}, p_{\phi})\). However, the explicit relationship of these parameters with system, noise, and disturbance characteristics such as \(n_{x}\), \(n_{u}\), \(n_{\phi}\), \(\sigma_{u}\), and \(\sigma_{w}\) is not known and we will address this in our future work. Consequently, we estimate these parameters numerically and utilize these estimates to calculate the theoretical rates discussed in Section~\ref{sec:num-exp}. While $b_{\phi}$ and $\bar{b}_{\phi}$ are straightforward to estimate, special attention is required to estimate the BMSB parameters. This section is dedicated to describing this estimation process.
 
For this, consider a system of the form~\eqref{eq:sys}. For this system, our goal is to estimate \( s_{\phi} \) and \( p_{\phi} \), where  
\[
    p_{\phi} = \inf_{\mathcal{F}_{t}, t \geq 0}\ \inf_{\|v\|_{2}=1}\  {\mathbb{P}}\bigg(|v^T\phi(z_{t+1})| \geq s_{\phi} \ \big|\ {\mathcal{F}}_{t}\bigg)
\]
numerically. First, observe that \( \phi(z_{t+1})\ |\ {\mathcal{F}}_{t} = \phi( \theta_{*}^{T}\phi(z_{t}) + w_{t}, u_{t+1} ) \), where \( z_{t} \in \mathcal{F}_{t} \). This implies that \( \phi(z_{t+1})\ |\ {\mathcal{F}}_{t} \) is a random variable influenced by \( w_{t} \) and \( u_{t+1} \). We proceed by fixing \( s_{\phi} = \bar{s} \) (for some $\bar{s} > 0$) and empirically estimate \( p_{\phi} \). To accomplish this, we first select a time horizon \( T \) and generate several trajectories of length \( T \) for the system. Let \( \mathcal{D}^{T} \) represent the set of these trajectories, while \( \mathcal{D}^{t} \) denotes a subset containing all trajectories up to \( t \leq T \). Additionally, we create multiple vectors \( v \in \mathbb{R}^{n_{\phi}} \) such that \( \|v\|_{2}=1 \); we refer to this collection as \( \bar{\mathcal{V}} \). These vectors are randomly sampled from a Gaussian distribution and subsequently normalized.

We then estimate \( \bar{p} \) as:
\[
\bar{p} = \min_{t \in [T]} \min_{z \in \mathcal{D}^{t}} \min_{v \in \bar{\mathcal{V}}}\  \mathbb{P} \bigg(|v^T\phi(\theta_{*}^{T}\phi(z) + w_{t}, u_{t+1})| \geq \bar{s} \bigg).
\]

As \( T \) increases, along with the number of trajectories and vectors \( v \), the minimum estimates will more accurately reflect the infimums. In this context, \( \bar{p} \) represents the minimum of \( \mathbb{P} \bigg(|v^T\phi(\theta_{*}^{T}\phi(z) + w_{t}, u_{t+1})| \geq \bar{s} \bigg) \) across all combinations in \( [T] \times \mathcal{D}^{T} \times \bar{\mathcal{V}} \). For each combination in this set, we estimate the probability \( \mathbb{P} \bigg(|v^T\phi(\theta_{*}^{T}\phi(z) + w_{t}, u_{t+1})| \geq \bar{s} \bigg) \) using Monte Carlo simulations. This process entails generating multiple random samples based on the distributions of \( w_{t} \) and \( u_{t+1} \), verifying whether each sample satisfies the condition \( |\bar{v}^T\phi(z)| \geq \bar{s} \), and tallying how many samples meet this criterion. We repeat this procedure for various values of \( \bar{s} \) until we identify a pair of \( (\bar{s}, \bar{p}) \) such that \( \bar{p} \in (0, 1) \). The estimated probability is calculated as the ratio of the count of successful samples to the total number of samples. According to the law of large numbers, this ratio converges to the true probability as the number of samples increases. For our estimations, we select \( T=50 \), \( |\bar{\mathcal{V}}|=1000 \), and \( |\mathcal{D}^{T}|=20 \).

\end{appendix}

% \newpage
% \input{checklist}

\end{document}